\documentclass[aps,10pt,groupedaddress,prx,superscriptaddress,twocolumn,balancelastpage]{revtex4-2}

\RequirePackage[english]{babel}
\RequirePackage[T1]{fontenc}
\RequirePackage[utf8]{inputenc}
\RequirePackage{amsmath}
\RequirePackage{amsthm}
\RequirePackage{amssymb}
\RequirePackage{amsfonts}
\RequirePackage{amsbsy}
\RequirePackage{colonequals}
\RequirePackage{nccmath}
\RequirePackage{mathtools}
\RequirePackage{bm}
\RequirePackage{dsfont}
\RequirePackage{lipsum}
\RequirePackage[hidelinks]{hyperref}
\RequirePackage{ragged2e}
\RequirePackage{textcomp}
\RequirePackage{gensymb}
\RequirePackage{braket}
\RequirePackage{xifthen}
\RequirePackage{makecell}
\RequirePackage[inline]{enumitem}
\RequirePackage{booktabs}
\RequirePackage{tikz}
\RequirePackage[mode=build]{standalone}
\RequirePackage{tcolorbox}
\RequirePackage{mdframed}
\RequirePackage{geometry}
\RequirePackage[createShortEnv]{proof-at-the-end}
\RequirePackage[classfont=sanserif,langfont=sanserif,funcfont=sanserif]{complexity}
\RequirePackage[ruled,lined,linesnumbered]{algorithm2e}
\SetKwComment{Comment}{$\triangleright$\ }{}
\RequirePackage{tabularx}
\newcolumntype{Y}{>{\centering\arraybackslash}X}

\RequirePackage{colortbl}
\RequirePackage{hhline}

\RequirePackage{booktabs}
\RequirePackage{pifont}
\RequirePackage[justification=justified,format=plain]{subcaption}

\RequirePackage{ytableau}
\ytableausetup{centertableaux}

\RequirePackage{multirow}

\RequirePackage{hhline}

\RequirePackage{physics}
\let\Tr\tr

\usetikzlibrary{arrows,3d,shapes,calc,decorations.pathreplacing,decorations.markings,positioning,intersections,shapes.symbols,positioning,calligraphy,matrix}

\tikzset{
	BC/.style = {decorate, 
		decoration={calligraphic brace, amplitude=5pt, raise=1mm},
		very thick, pen colour={black}
	},
}

\DeclareMathAlphabet{\mymathbb}{U}{BOONDOX-ds}{m}{n}

\theoremstyle{plain}
\newtheorem{theorem}             {Theorem}
\newtheorem{proposition}{Proposition}

\newtheorem{lemma}{Lemma}

\newtheorem*{theorem*}    {Theorem}
\newtheorem*{proposition*}{Proposition}
\newtheorem*{lemma*}      {Lemma}
\newtheorem*{corollary*}  {Corollary}
\newtheorem*{conjecture*} {Conjecture}

\theoremstyle{definition}
\newtheorem{definition}{Definition}

\newtheorem*{definition*}{Definition}
\newtheorem*{example*}   {Example}

\newtheoremstyle{applicationstyle} 
    {0em}                    
    {0em}                    
    {}                   
    {\parindent}                           
    {\bfseries}                   
    {.}                          
    {3pt}                       
    {#1 #2: #3}  

\theoremstyle{applicationstyle}
\newtheorem{application}{Application}

\numberwithin{application}{subsection} 
\newtcolorbox{mybox}[2][]{
               = {yshift=-8pt},
  colback      = cyan!6!white,
  colframe     = cyan!1!black,
  halign       = flush left,
  fonttitle    = \bfseries\sffamily,
  colbacktitle = cyan!50!black,
  title        = #2,#1,
  }

\newcommand{\phsp}[2]{%
\ifthenelse{\isempty{#1}}
	{\rule[-.5\baselineskip]{0pt}{.5\baselineskip}}%
	{\rule[#1\baselineskip]{0pt}{#2\baselineskip}}%
}

\newcommand{\Id}{\mathbb{I}}

\newcommand{\intg}[3]{\int_{#1}\mathrm{d}#3 #2}

\newcommand{\aco}{\left\{}
\newcommand{\acf}{\right\}}
\newcommand{\lb}{\left|}
\newcommand{\rb}{\right|}
\newcommand{\lpr}{\left(}
\newcommand{\rpr}{\right)} 
\newcommand{\lbr}{\left[}
\newcommand{\rbr}{\right]} 
\renewcommand{\abs}[1]{\lb\!#1\rb}

\makeatletter
\newcommand{\leqnomode}{\tagsleft@true\let\veqno\@@leqno}
\newcommand{\reqnomode}{\tagsleft@false\let\veqno\@@eqno}
\makeatother

\makeatletter
\newcommand{\proglabel}[2]{%
   \protected@write \@auxout {}{\string \newlabel {#1}{{#2}{\thepage}{#2}{#1}{}} }%
   \hypertarget{#1}{#2}
}
\makeatother

\newcommand{\Ac}{\mathcal{A}}
\newcommand{\Bc}{\mathcal{B}}

\newcommand{\Dc}{\mathcal{D}}

\newcommand{\Fc}{\mathcal{F}}

\newcommand{\Hc}{\mathcal{H}}

\newcommand{\Kc}{\mathcal{K}}

\newcommand{\Mc}{\mathscr{M}}
\newcommand{\Oc}{\mathcal{O}}

\newcommand{\Pc}{\mathcal{P}}
\newcommand{\Uc}{\mathcal{U}}

\newcommand{\Xc}{\mathcal{X}}
\newcommand{\Yc}{\mathcal{Y}}
\newcommand{\Zc}{\mathcal{Z}}
\newcommand{\Ls}{\mathscr{L}}
\newcommand{\Bs}{\mathscr{B}}

\newcommand{\Ms}{\mathscr{M}}

\newcommand{\Vs}{\mathscr{V}}

\newcommand{\U}{\mathbb{U}}

\newcommand{\N}{\mathbb{N}}

\renewcommand{\S}{\textsf{S}}
\usepackage{mathrsfs}

\makeatletter
\renewcommand{\fnum@figure}{Figure \thefigure}
\renewcommand{\fnum@table}{Table \thetable}
\makeatother

\renewcommand{\thetable}{\arabic{table}}

\usepackage[nameinlink]{cleveref} 
\AddToHook{cmd/appendix/before}{\crefalias{section}{appendix}}
\crefname{chapter}{Chapter}{Chapters}
\crefname{section}{Section}{Sections}
\crefname{algorithm}{Algorithm}{Algorithms}
\crefname{application}{Application}{Applications}
\crefname{line}{Line}{Lines}
\crefname{equation}{Equation}{Equations}
\crefname{figure}{Figure}{Figures}
\crefname{subfigure}{Panel}{Panels}
\crefname{table}{Table}{Tables}
\crefname{appendix}{Appendix}{Appendices}
\crefname{theorem}{Theorem}{Theorems}
\crefname{corollary}{Corollary}{Corollaries}
\crefname{lemma}{Lemma}{Lemmas}
\crefname{proposition}{Proposition}{Propositions}
\crefname{definition}{Definition}{Definitions}
\crefname{footnote}{Footnote}{Footnotes}

\makeatletter
\renewcommand\p@subfigure{}

\DeclareCaptionLabelFormat{mysublabelfmt}{(\alph{sub\@captype})}
\makeatother
\let\cref\cref

\newcommand{\bs}[1]{\boldsymbol{#1}}
\newcommand{\bigo}[1]{\Oc\!\left(#1\right)}

\newcommand{\pr}[2][]{
	\mathop{
		\ifx &#1&
		\mathrm{Pr}
		\else
			\mathrm{Pr}_{#1}
		\fi
		\left[#2\right]}
}

\newcommand{\e}[2][]{
	\mathop{
		\ifx &#1&
			\mathbb{E}
		\else
			\underset{#1}{\mathbb{E}}
		\fi
		\left[#2\right]}
}

\renewcommand{\var}[1]{\mathrm{Var}\left[#1\right]}

\newcommand{\s}{{\bm{s}}}

\newcommand{\n}{{\bm{n}}}

\newcommand{\x}{{\bm{x}}}
\newcommand{\y}{{\bm{y}}}
\renewcommand{\t}{{\bm{t}}}

\def\multiset#1#2{\ensuremath{\left(\kern-.3em\left(\genfrac{}{}{0pt}{}{#1}{#2}\right)\kern-.3em\right)}}

\usepackage{pgfplots}
\pgfplotsset{compat=newest}
\usepgfplotslibrary{groupplots}
\usepgfplotslibrary{polar}
\usepgfplotslibrary{smithchart}
\usepgfplotslibrary{statistics}
\usepgfplotslibrary{dateplot}
\usepgfplotslibrary{ternary}
\usetikzlibrary{arrows.meta}
\usetikzlibrary{backgrounds}
\usepgfplotslibrary{patchplots}
\usepgfplotslibrary{fillbetween}
\pgfplotsset{%
    layers/standard/.define layer set={%
        background,axis background,axis grid,axis ticks,axis lines,axis tick labels,pre main,main,axis descriptions,axis foreground%
    }{
        grid style={/pgfplots/on layer=axis grid},%
        tick style={/pgfplots/on layer=axis ticks},%
        axis line style={/pgfplots/on layer=axis lines},%
        label style={/pgfplots/on layer=axis descriptions},%
        legend style={/pgfplots/on layer=axis descriptions},%
        title style={/pgfplots/on layer=axis descriptions},%
        colorbar style={/pgfplots/on layer=axis descriptions},%
        ticklabel style={/pgfplots/on layer=axis tick labels},%
        axis background@ style={/pgfplots/on layer=axis background},%
        3d box foreground style={/pgfplots/on layer=axis foreground},%
    },
}

\renewcommand{\geq}{\geqslant}
\renewcommand{\leq}{\leqslant}

\definecolor{lightblue}{RGB}{84,189,220}
\definecolor{navy}{RGB}{47,60,126}
\definecolor{quandelablue}{HTML}{435BEC}
\definecolor{quandeladarkblue}{HTML}{224C98}
\definecolor{quandelared}{HTML}{F07362}
\definecolor{darkviolet}{RGB}{99,56,142}
\definecolor{darkgreen}{RGB}{39,174,96}
\definecolor{violet}{RGB}{222,49,99}
\definecolor{lightblue}{RGB}{84,189,220}
\definecolor{navy}{RGB}{47,60,126}
\definecolor{vertsauj}{HTML}{D4E3CF}
\definecolor{oranj}{HTML}{FFAB0D}

\newclass{\BosonPCC}{BosonP}

\newcommand{\per}[1]{\mathrm{Per}\left(#1\right)}

\renewcommand{\S}{S}

\RequirePackage{anyfontsize}

\newcommand{\VS}{\text{VS}}

\newcommand{\Minv}[2]{\Mc^{(#1)^{-1}}\!\!\lpr#2\rpr}
\newcommand{\mdiag}{\mathfrak{D}}
\let\gt\mathfrak

\newcommand{\psn}{\mathscr{P}}

\makeatletter
\def\section{%
  \@startsection{section}{1}{\z@}%
    {0.4cm \@plus 1ex \@minus .2ex}
    {0.2cm}
    {\centering\normalfont\small\bfseries}%
}
\def\subsection{%
  \@startsection{subsection}{2}{\z@}%
    {0.4cm \@plus 1ex \@minus .2ex}%
    {0.15cm}%
    {\centering\normalfont\small\bfseries}%
}
\def\subsubsection{%
  \@startsection{subsubsection}{3}{\z@}%
    {0.25cm \@plus 1ex \@minus .2ex}%
    {0.12cm}%
    {\centering\normalfont\small\itshape}%
}
\makeatother

\makeatletter

\newif\ifinappendix
\inappendixfalse

\pretocmd{\appendix}{\inappendixtrue}{}{}

\let\orig@section\section
\let\orig@subsection\subsection
\let\orig@subsubsection\subsubsection

\renewcommand{\section}{%
  \@ifstar
    {\orig@section*}%
    {\app@section}%
}

\newcommand{\app@section}[1]{%
  \orig@section{#1}%
  \ifinappendix
    \addcontentsline{atoc}{section}{%
      \protect\numberline{\thesection}#1%
    }%
  \fi
}

\renewcommand{\subsection}{%
  \@ifstar
    {\orig@subsection*}%
    {\app@subsection}%
}

\newcommand{\app@subsection}[1]{%
  \orig@subsection{#1}%
  \ifinappendix
    \addcontentsline{atoc}{subsection}{%
      \protect\numberline{\thesubsection}#1%
    }%
  \fi
}

\renewcommand{\subsubsection}{%
  \@ifstar
    {\orig@subsubsection*}%
    {\app@subsubsection}%
}

\newcommand{\app@subsubsection}[1]{%
  \orig@subsubsection{#1}%
  \ifinappendix
    \addcontentsline{atoc}{subsubsection}{%
      \protect\numberline{\thesubsubsection}#1%
    }%
  \fi
}

\makeatother

\newcommand{\arxiv}[2]{#1}

\hypersetup{
    colorlinks=true,
    linkcolor=navy,
    citecolor=navy,
    urlcolor=navy,
    breaklinks=true,
}

\begin{document}

\title{Learning photonic quantum states}
\author{Hugo Thomas}
\thanks{hugo.thomas@quandela.com}
\affiliation{Quandela, 7 rue Léonard de Vinci, 91300 Massy, France}
\affiliation{Sorbonne Université, CNRS, LIP6, 75005 Paris, France}
\affiliation{DIENS, \'Ecole Normale Supérieure, PSL University, CNRS, INRIA, 45 rue d'Ulm, 75005 Paris, France}
\author{Ulysse Chabaud}
\affiliation{DIENS, \'Ecole Normale Supérieure, PSL University, CNRS, INRIA, 45 rue d'Ulm, 75005 Paris, France}
\author{Pierre-Emmanuel Emeriau}
\affiliation{Quandela, 7 rue Léonard de Vinci, 91300 Massy, France}

\date{\today}

\begin{abstract}

Learning quantum state properties is both a fundamental and practical problem in
quantum information theory. Classical shadows have emerged as an efficient
method for estimating properties of unknown quantum states, with rigorous
statistical guarantees, by performing randomized measurement on few copies of
the state. With the advent of photonic technologies, formulating efficient
learning algorithms for such platforms comes out as a natural problem. Here, we
introduce a practical classical shadow protocol for learning photonic quantum
states via randomized passive linear optical transformations and photon-number
measurement. We provide rigorous theoretical guarantees showing that our scheme
is sample- and time-efficient for measuring physical observables of interest. We
experimentally demonstrate our photonic classical shadow protocol on both a
twelve-mode and a twenty-four-mode integrated quantum processing unit, and
showcase its versatility with five different applications, including Hamiltonian
measurement and learning complex photonic states.

\end{abstract}

\maketitle

\section{Introduction}

\noindent Characterizing unknown quantum states is a fundamental challenge in
quantum information science. Quantum state tomography reconstructs the
full density matrix of a state $\rho$, but suffers from exponential scaling with the system size. As ever-bigger experimental systems become accessible thanks to
recent quantum hardware development
\cite{nigmatullin_experimental_2025,larsen_integrated_2025,liu_robust_2025,salesrodriguez_experimental_2025},
full-fledged quantum state tomography rapidly becomes intractable, since even
storing the density matrix in classical memory is hopeless. The question
is no longer on how to reconstruct the full state, but rather about how much
useful information can be extracted from a limited number of copies.

A breakthrough in this direction came with shadow tomography.
Aaronson~\cite{aaronson_shadow_2018} first showed that  predicting physical
properties, reconstructing the entire density matrix is unnecessary. Instead, a
well-chosen, lightweight description of the quantum state---a \textit{shadow} of
the state---suffices. Huang et al.~\cite{huang_predicting_2020} subsequently
introduced the notion of \textit{classical shadows} produced from randomized
measurements on only a few copies of the state and classical post-processing.
These shadows enable the efficient prediction of a wide variety of quantum
features, including expectation values, fidelities, and entropies. Since then,
the framework has been widely adopted and refined, incorporating error
mitigation~\cite{jnane_quantum_2024,hu_demonstration_2025,chen_robust_2021},
adaptation to hardware
constraints~\cite{west_real_2025,koh_classical_2022,zhao_fermionic_2021,wan_matchgate_2023,sauvage_classical_2024,ippoliti_classical_2024,low_classical_2024},
including continuous-variable bosonic systems
\cite{gandhari_precision_2024,becker_classical_2024} and improvements of the
protocol
itself~\cite{nguyen_optimizing_2022,helsen_thrifty_2023,zhou_performance_2023,chen_adaptivity_2024,fawzi_learning_2024}.
As a result, classical shadows have emerged as a practical tool across domains,
from
benchmarking~\cite{levy_classical_2024,kunjummen_shadow_2023,zhu_crossplatform_2022,elben_crossplatform_2020}
to error correction \cite{conrad_chasing_2026}, Hamiltonian
simulation~\cite{hadfield_measurements_2022,mcnulty_estimating_2023,dutt_practical_2023,huang_provably_2022}
and quantum machine
learning~\mbox{\cite{huang_power_2021,jerbi_shadows_2024,haug_quantum_2023}}.

Despite this prompt progress, most developments of shadow tomography have
focused on qubit-based architectures. Yet, photonics has emerged as one of the
most promising candidates for scalable quantum information
processing~\cite{knill_scheme_2001,bartolucci_fusionbased_2023,degliniasty_spinoptical_2024}.
Photonic qubits are inherently modular and naturally suited for networking, as
they can be routed, multiplexed and distributed with minimal cross-talk
\cite{wang_integrated_2020}. Integrated photonics offers a path to large-scale
quantum technologies by leveraging semiconductor fabrication techniques.
Moreover, NISQ photonic devices are already publicly
accessible~\cite{maring_versatile_2024}, and photonics is a strong candidate for
quantum advantage demonstration through boson
sampling~\cite{aaronson_computational_2011}. Experimental realizations of boson
sampling have been reported with increasing scale and
sophistication~\cite{spring_boson_2013,broome_photonic_2013,loredo_boson_2017,wang_highefficiency_2017,tillmann_experimental_2013,hoch_reconfigurable_2022,hoch_quantum_2025},
motivating algorithms that explicitly exploit the hardness of simulating linear
optics~\cite{salavrakos_photonnative_2025}. However, throughout these advances
in photonic quantum information, characterizing multi-photon states in many
modes remains a central bottleneck. While randomized benchmarking protocols have
been introduced to assess the gate quality of discrete-variable bosonic systems
\cite{arienzo_bosonic_2025,wilkens_benchmarking_2024}, full tomography remains
prohibitively demanding~\cite{banchi_multiphoton_2018}.

Here, we bridge this gap by introducing a classical shadow protocol tailored to
photonic platforms in a practically motivated setting. This setting consists of
passive linear optics combined with photon-number resolving (PNR) detectors, an
architecture that is both experimentally mature and conceptually distinct from
qubit systems. This setting imposes two constraints: transformations are
limited to passive linear-optical transformations and measurements are restricted to photon-number
detection. At first sight, these restrictions appear to preclude the versatility
of shadow tomography.

\begin{figure*}[t]
    \includegraphics[width=\textwidth]{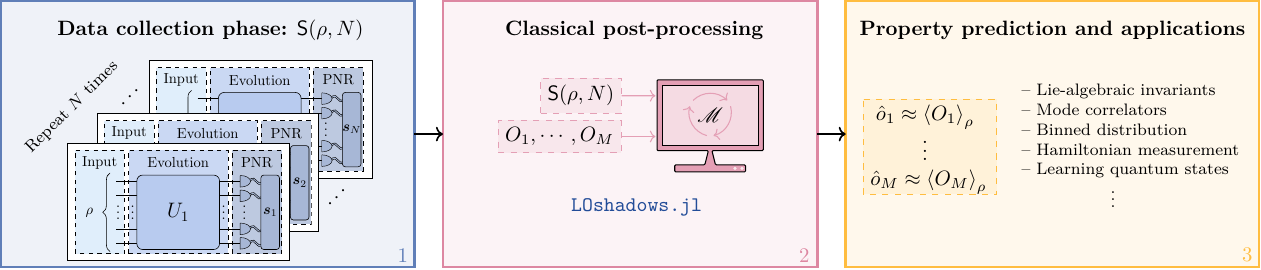}
    \caption{\justifying \textbf{Classical shadow pipeline for learning photonic
    states.} Illustration of the photonic classical shadow protocol. The classical shadow
    $\textsf{S}(\rho, N)$ consists of all pairs $(U_i, \bm s_i)$ of random
    linear optical network and associated measurement outcome obtained during
    the data collection phase. We provide a Julia package
    \cite{github_loshadows} implementing the post-processing required to
    estimate properties of the unknown input state; details can be found in
   \cref{app:exp-meth}.\label{fig:photonicShadows}}
\end{figure*}

We show that scalable classical shadow protocols remain feasible under these
conditions. The main obstacle is PNR detection which erases coherence between
states of different photon-number, restricting tomography to individual
photon-number subspaces~\cite{banchi_multiphoton_2018}. However, this matches
most photonic applications, which operate at fixed photon-number, including
linear optical operations and dual-rail qubit encodings for universal
computation. 
Crucially, our protocol applies to generic photonic input states including those produced by nonlinear processes. 
This means that a non-universal, linear-optical
device suffices to characterize highly complex states produced by a universal
photonic quantum computer such as that envisioned in the Knill--Laflamme--Milburn scheme \cite{knill_scheme_2001} or in fusion networks \cite{bartolucci_creation_2021}. 
We obtain theoretical bounds on both the sample- and
time-complexity required for estimating a wide class of properties of practical
interest of an unknown photonic quantum state.
Finally, we
demonstrate experimentally the wide versatility of our protocol on five distinct tasks using a twelve-mode and
a twenty-four-mode linear optical quantum computer with up to four
single-photons. We show that many physical properties of a photonic state can be
efficiently and accurately predicted, with applications to benchmarking,
certification and learning among others, as pictured in
\cref{fig:photonicShadows}. Notably we provide the first experimental implementation of a Boson Sampling state learning task as outlined in \cite{iosue_higher_2025}.

\section{Photonic classical shadows}

\textbf{Background.} Classical shadows \cite{huang_predicting_2020} provide
lightweight classical descriptions of an unknown quantum state $\rho$, enabling
the estimation of a collection of properties of that state: linear properties,
such as expectation values of quantum observables, or nonlinear properties, such
as sub-system entropies. The protocol consists in three steps: 1) a data
acquisition stage, 2) a classical post-processing stage and 3) a property
prediction stage. 

The data-acquisition stage relies on the following measurement primitive: a
random unitary evolution $U$ drawn uniformly at random from a set $\Uc$ is applied to
the input state, followed by a projective measurement in the Fock basis
$\Bc$ yielding $b$. A snapshot corresponds to the classical description
$U^\dagger \ketbra{b} U$. The process of repeatedly applying the measurement
primitive may be seen as a quantum channel 
\begin{equation}
    \Mc(\rho) = \e[\substack{U \sim \Uc \\ b \sim \Dc_U}]{U^\dagger \ketbra{b} U},
\end{equation}
where $\Dc_U$ denotes the probability distribution induced by the Born rule on
the post-evolution state $U\rho U^\dagger$. The amount of information that can
be recovered from a classical shadow is described by the \emph{visible space}
\cite{kirk_hardwareefficient_2022} of the channel. More precisely, the visible
space of $\Mc$ is the span of $U^\dagger\ketbra{b}U$ for every ${U \in \Uc}$ and
${b \in \Bc}$. Finally, by linearity of the trace and the expectation value,
linear functions $O$ in the visible space satisfy
\begin{equation}
    \expval{O}_{\rho} = \e[U, b]{\expval{O}_{\Mc^{-1}(U^\dagger\ketbra{b}U)}}.
\end{equation}

Our photonic classical shadow protocol is defined in the following context. The
space of $m$-mode photonic states is the infinite-dimensional Fock-Hilbert space
$\smash{\Fc_m = \bigoplus_{n \geq 0} \Hc_m^n}$, where $\Hc_m^n$ is the
$n$-photon subspace of dimension $\binom{n+m-1}{n}$ spanned by $m$-tuples of
nonnegative integers summing to $n$, which we denote $\Phi_m^n$.
Linear-optical transformations acting on $m$ bosonic modes are described by the
unitary group $\U(m)$. A transformation $U \in \U(m)$ acts linearly on the mode
operators and induces, via second quantization, a unitary representation
$\varphi_m$ of $\U(m)$ on $\Fc_m$ \cite{aniello_exploring_2006}. 
Linear-optical transformations preserve the total photon-number, $\varphi_m$
decomposes into irreducible representations (\emph{irreps} for short) living in
subspaces characterized by fixed photon-number.

\medskip

\textbf{Protocol.} We now describe how to perform shadow tomography of an
unknown $m$-mode Fock state $\rho$, i.e., how to produce a classical
representation $\hat\rho$ that behaves like $\rho$ on average. The only a priori
knowledge we assume is $m$, the number of modes. Our protocol only requires
passive linear-optical transformations and PNR measurement and applies to
generic photonic quantum states (not necessarily prepared by passive linear
optical transformations). In this setting, we show that the visible space is
delineated by the measurements that project onto a fixed photon-number subspace.

The data collection process is similar to the classical shadow protocol for
qubits. A snapshot is obtained by drawing a unitary matrix $U \sim \mu_H$ at
random from the Haar measure (which can easily be done numerically
\cite{mezzadri_how_2007}), applying the transformation ${\omega_m(U):\rho
\mapsto \varphi_m(U)\,\rho\,\varphi_m^\dagger(U)}$ by letting $\rho$ evolve
through the linear network described by $U$, and measuring in the computational
basis with PNR detection. Upon measuring a total of $n$ photons, the measurement
outcome is $\bs s \in \Phi_m^n$ with probability
$\tr{\omega_m(U)(\rho)\ketbra{\s}}$.

The protocol therefore produces a description of an arbitrary input state in the
form of a convex combination of unitarily evolved Fock basis states. Unlike
Clifford or Pauli shadows for qubits, efficiently computing basis state evolution
is not possible in linear optics in general. To this end, we call a classical
shadow of the photonic state $\rho$ the collection
\begin{equation}
    \S(\rho, N) = \{(U_1, \bm s_1), \cdots,(U_N, \bm s_N)\}.
\end{equation}

The average mapping (over the choice of unitary transformations and measurement
outcome) of the input state $\rho$ naturally takes the form of a quantum channel
$\Mc$. As unitary transformations preserve the number of photons and
photon-number detection projects onto one of the subspaces of $\Fc_m$, the
measurement channel is not tomographically complete. In particular, applying
$\Mc$ yields a block diagonal operator, i.e., we have
\begin{equation}\label{eq:channel}
    \Mc = \bigoplus_{n \geq 0}{\Mc^{(n)}} \circ \Pc^{(n)},
\end{equation}
where $\Pc^{(n)}$ projects onto the $n$-photon subspace $\Hc_m^n$. Using the
decomposition of $\omega_m$ into irreducible representations, we give in
\arxiv{\cref{sec:channel}}{Supplementary Information} a closed form for the channel, whose derivation
follows from Schur's lemma.
It follows that the visible space $\Vs$ associated with the channel of
\cref{eq:channel} also has a block-diagonal structure, more precisely ${\Vs =
\bigoplus_{n \geq 1} \Ls(\Hc_m^n)}$. 
In the remainder, all observables $O$ we consider belong to the visible space
$\Vs$. Using the self-adjointness property of the measurement channel, we find
that
\begin{equation}\label{eq:trORhShadows}
    \expval{O}_{\rho}=
    \e[\substack{U \sim \mu_{H} \\ \s \sim \Dc_U}]{\bra{\s}\varphi_{m}(U)\Mc^{-1}(O)\varphi_m^{\dagger}(U)\ket{\s}}\!,
\end{equation}
where $\Mc^{-1}$ denotes the inverse of $\Mc$ viewed as a linear map. Hence,
linear functions of an unknown quantum state $\rho$ can be learned by averaging
the elements of a photonic classical shadow $\textsf{S}(\rho, N)$. 

\medskip

\textbf{Sample complexity.} 
Next, we provide a rigorous upper bound on the number of samples for
guaranteeing a desired precision and a given confidence. 

To that end, akin to the qubit shadow-norm \cite{huang_predicting_2020}, we
introduce the \emph{photonic shadow-norm} and denote it $\|\cdot\|_{\psn}$, see
\arxiv{\cref{sec:varianceBound}}{Supplementary Information} for a formal
definition. Importantly, we show that for all observables $O$ that are
polynomials in the creation and annihilation operators, the photonic shadow-norm
is bounded by a function of their degree as 
\begin{equation}
        \| O \|_{\psn}^2 =  O(\|O\|_{\infty}^2m^{3\deg(O)}).
\end{equation}
It follows that the sample complexity of the scheme is dictated by the degree of
the observables of interest. This is summarized by the following result \arxiv{(see the
\cref{sec:estExpVal} for a formal generalization)}{}:

\begin{theorem}[Sample complexity of photonic classical
    shadows]\label{thm:informalSample} A collection of $T$ linear functions
    $\expval{O_1}_{\rho}, \dots, \expval{O_T}_{\rho}$ can be collectively
    estimated to within additive precision $\varepsilon$ with a classical shadow
    of size ${O(\max_t\|O_{t_0}\|_{\psn}^2\log (T) / \varepsilon^2)}$ with
    constant success probability, where $O_{t_0} = O_{t} -
    \frac{\tr{O_t}}{\dim{\Hc_m^n}}\Id$ is the traceless part of the observable.
\end{theorem}

We provide a detailed analysis in \arxiv{\cref{sec:estExpVal}}{the Supplementary Information} showing that
learning low-degree observable features of photonic quantum state is remarkably
efficient using representation-theoretic tools. This analysis is supported by
experimental results exposed in \cref{sec:experiments}. 

Notably, this generalises the locality condition for Pauli-based qubit classical
shadows \cite{huang_predicting_2020}, as low-weight Pauli observables in
dual-rail qubit encoding are low-degree monomials in the creation and
annihilation operators \cite{knill_scheme_2001}.

\medskip

\textbf{Time complexity.} Beyond their sample complexity, classical shadows
protocols also rely on heavy classical post-processing
\cite{huang_predicting_2020}. As such, it is important to ensure that the
corresponding running time is not prohibitive for practical applications. The
standard classical shadow protocol requires a post-processing that would not be
efficient in our setting, as it amounts to computing matrix permanents
\cite{aaronson_computational_2011}. Rather, we introduce in the
\arxiv{\cref{app:exactTechnique}}{Supplementary Information} a novel
technique for computing $\expval{O}_{\varphi_m(U)\ket{\bm s}}$ from the pair
$(U, \bm s)$ in time $\smash{\bigo{m^{\deg(O)}}}$, allowing for efficient
post-processing.

The computational complexity is summarized by the following result, which shows
in particular that our protocol is also computationally efficient for learning
low-degree observables:

\begin{theorem}[Computational complexity of photonic classical
    shadows]\label{thm:informalComp} Given a classical shadow $\S(\rho^{(n)},
    N)$ of an $m$-mode $n$-photon state, a linear function
    $\expval{O}_{\rho^{(n)}}$ can be computed exactly in time $\bigo{\poly(N, n,
    m^{\deg(O)})}$.
\end{theorem}

We note that while \cref{thm:informalSample,thm:informalComp} provide
pessimistic upper bounds on the sample and time complexities of photonic
classical shadows, in practice, the sample complexity and running time of the
protocol can be significantly lower, as we demonstrate experimentally in the
next section. Additionally, note that these bounds still ensure efficiency for
learning observable features of constant degree.

\section{Experiments}\label{sec:experiments}

\begin{figure*}
    \centering

    \subfloat[Two-mode correlators.\label{fig:expTPC}]{
        \includegraphics[width=.9\textwidth]{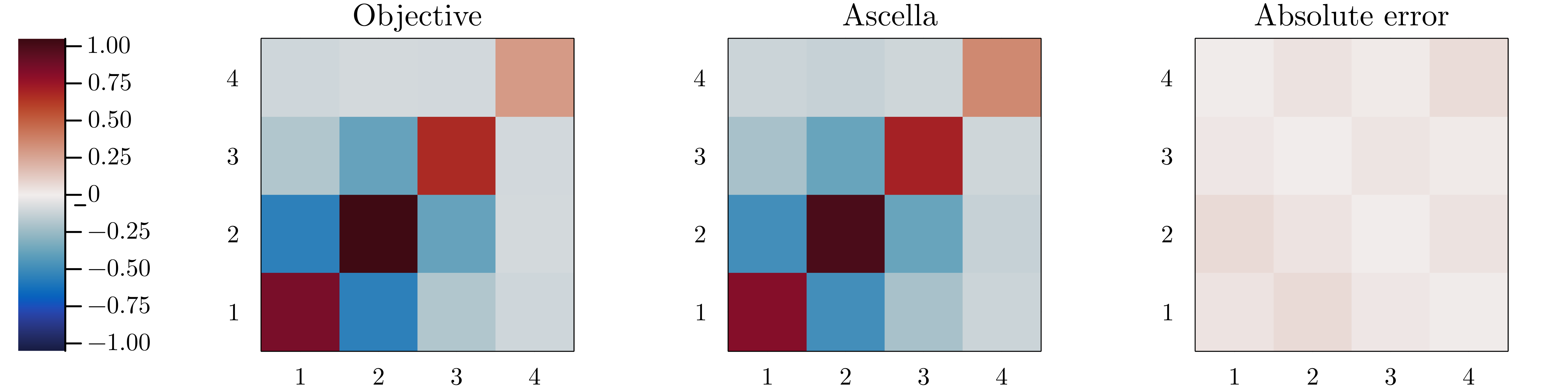}
    }

    \subfloat[Lie-algebraic linear optical invariants.\label{fig:expINV}]{
        \includegraphics[width=.9\textwidth]{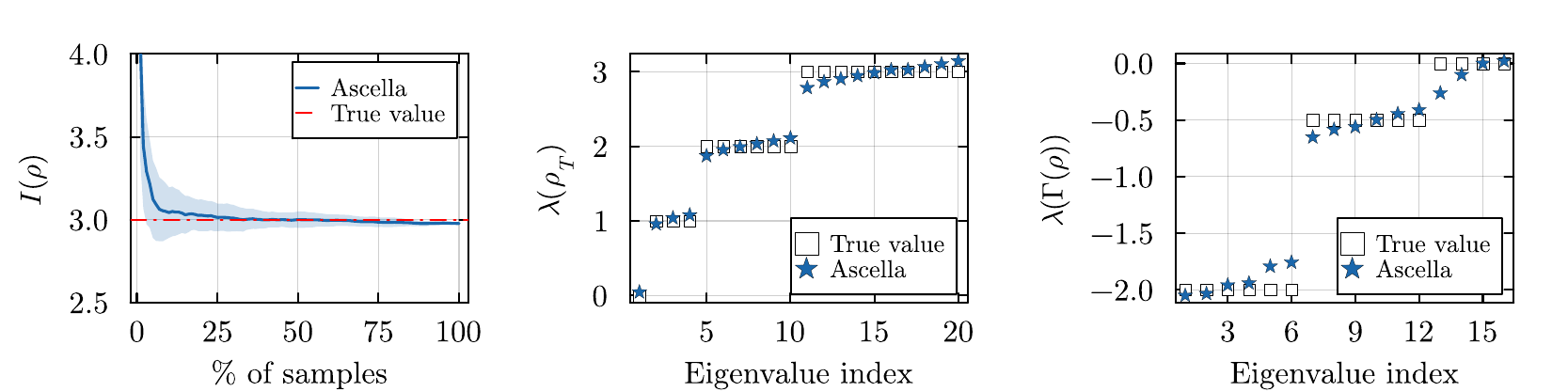}
    }

    \subfloat[Binned probability distributions.\label{fig:expBPB}]{
        \includegraphics[width=.9\textwidth]{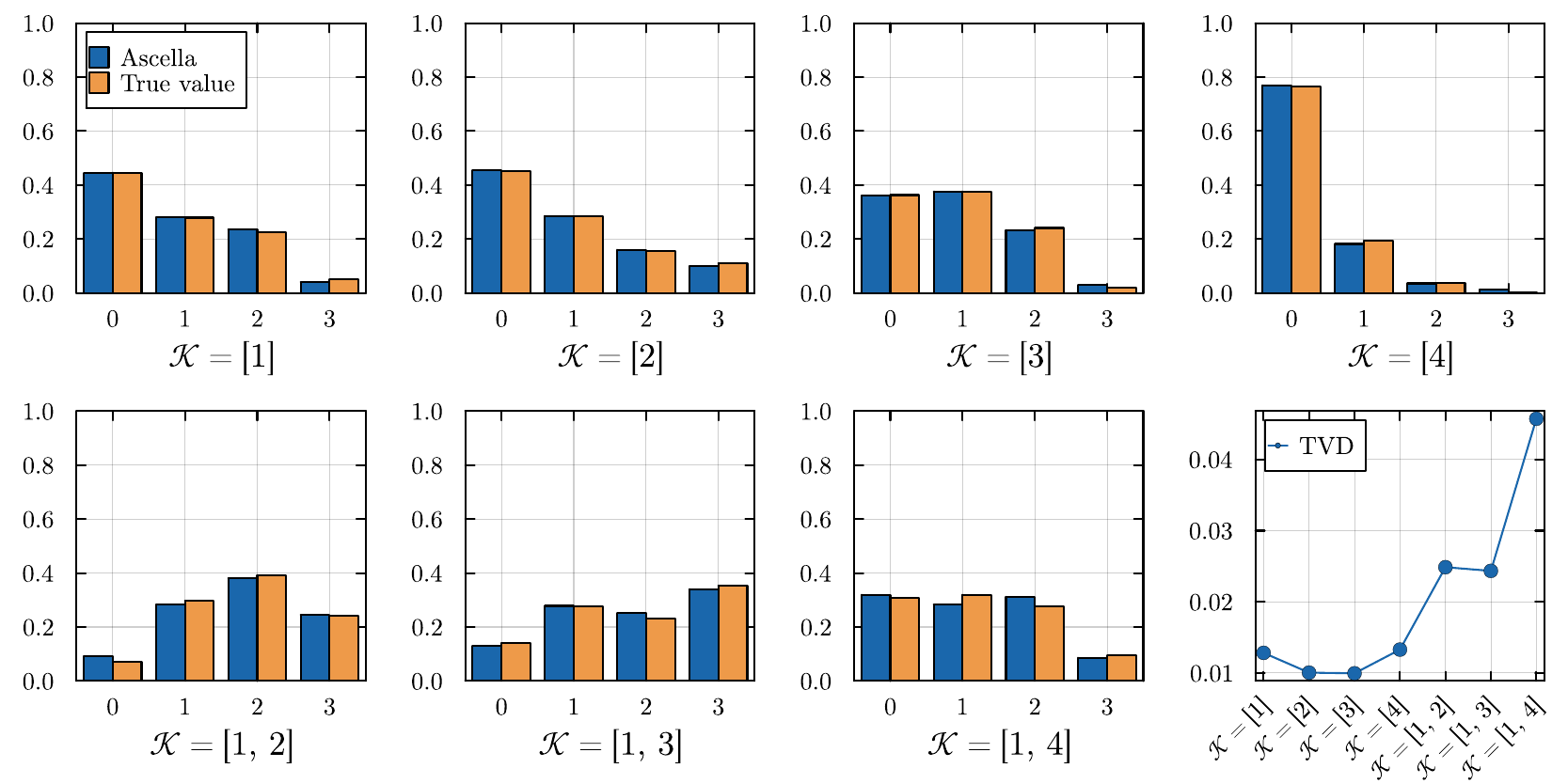}
    }

    \caption{
        \justifying
        \textbf{Experimental property estimation via classical shadows on
        \emph{Ascella}.} Experimental results of
        \cref{exp:loc,exp:inv,exp:bpb} performed on \emph{Ascella} are
        documented in \cref{fig:expTPC,fig:expINV,fig:expBPB}, respectively.
        The classical shadow of the input is the same for all three experiments.
        It consists of $N=1100$ randomly sampled $U_i$; for each of which an
        average of 19 3-photon samples were obtained. \cref{fig:expTPC}:
        Comparison of the two-mode correlation matrices of the true one and that
        obtained from the shadow. (left) true correlation matrix (center)
        estimate obtained from the shadows (right) entry-wise absolute
        estimation error. \cref{fig:expINV}: Numerical values of the
        Lie-algebraic invariants. (left) evolution of the estimation via
        bootstrapping, (center and right) estimated spectrum of $\rho_T$ and
        $\Gamma(\rho)$. \cref{fig:expBPB}: Binned-distribution for all
        possible bipartitions of the modes. This way, each bipartition is
        defined by a subset of the modes and its complement. For instance, $\Kc
        = [1]$ represents the partition $(\{1\}, \{2, 3,4\})$. Certain bins are
        omitted due to symmetry. The horizontal axis corresponds to the possible
        occupations of the labelled bin.}
    \label{fig:experimentsAscella}
\end{figure*}

\begin{figure*}
    \subfloat[Ground energy estimation.\label{fig:expGEE}]{%
     \includegraphics[width = \linewidth]{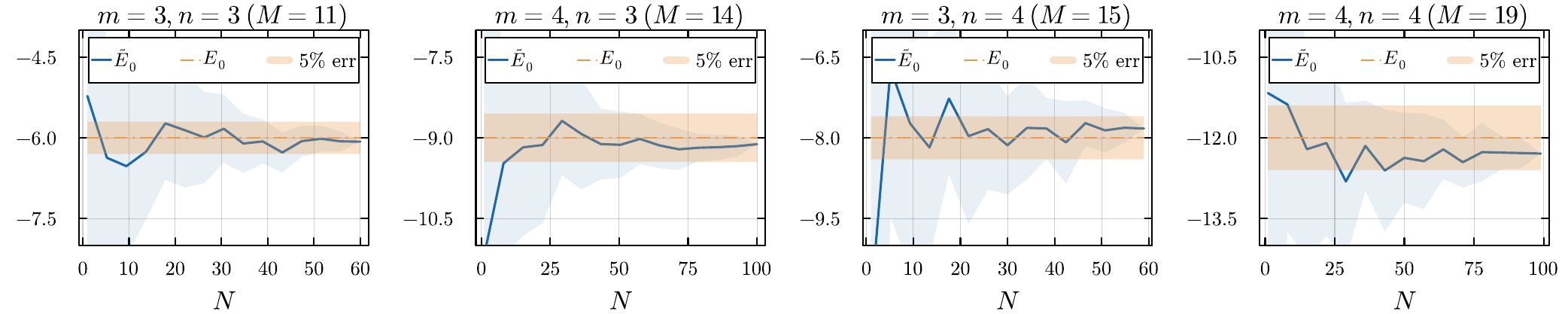}    
    }

    \subfloat[Learning Boson Sampling states.\label{fig:expBSL}]{%
        \includegraphics[width = .85\linewidth]{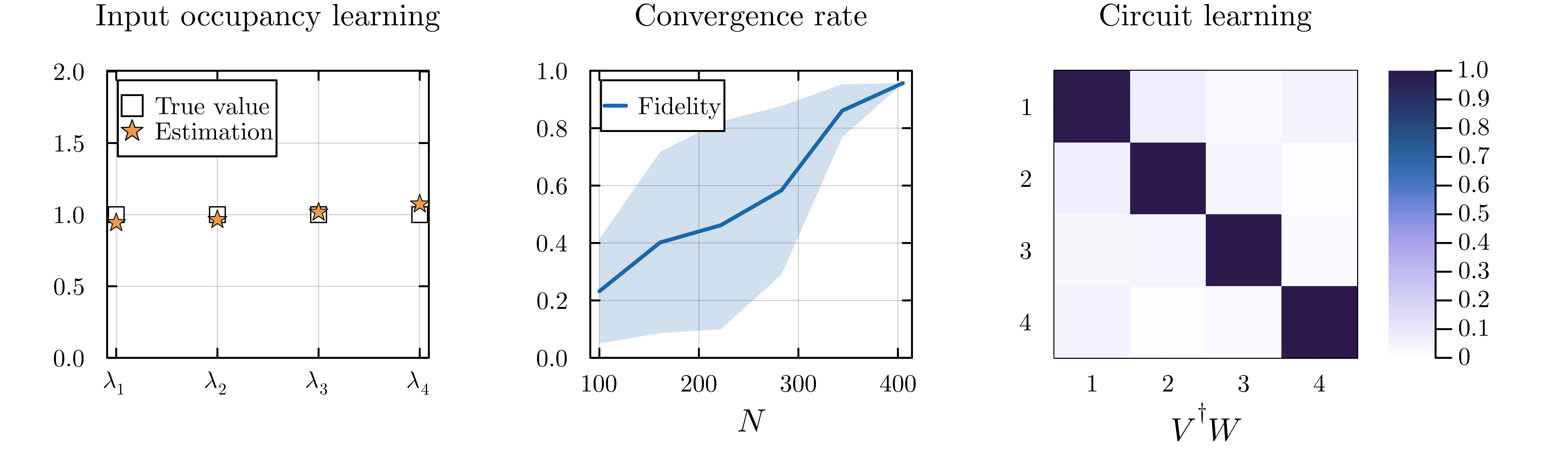}    
    }
   
    \caption{\justifying \textbf{Experiments on \emph{Belenos}.} Experimental
    results of \cref{exp:gee,exp:bsl} performed on \emph{Belenos}.
    \cref{fig:expGEE} We prepare the ground state of the $m = 3, 4$--site
    superfluid Bose-Hubbard Hamiltonian with $n = 3, 4$ photons, and estimate
    the associated energy using shadows of size $60$ and $100$ for $m=3$ and $4$
    respectively. The state preparation and measurement induce an overhead in
    the effective number of mode $M$ that are required, see Methods. Here, we
    use up to 19 modes of the photonic processor. \cref{fig:expBSL} We
    consider the four-photon input state ${\ket{\bm n} = \ket{1, 1, 1, 1}}$ and
    a pre-selected randomly chosen unitary matrix $W$. We collect a classical
    shadow of $\varphi_4(W)\ket{1,1,1,1}$ of size 405, and plot (left) the
    estimate of $\ket{\n}$ obtained by rounding the eigenvalues of a matrix
    constructed from first moments, (center) the evolution of the computed
    fidelity of the state obtained from a shadow of size $N$, and (right)
    $V^\dagger W$, that shows how close the learned circuit $V$ is from the
    ground truth.}
    \label{fig:experimentBelenos}
\end{figure*}

\noindent We report experimental implementations of the classical shadow
protocol based on passive linear optical transformations and photon-number
resolving measurements for predicting properties of the input state. 

We use two quantum processing units (QPUs): \emph{Ascella}
\cite{maring_versatile_2024} (QPU A) and \emph{Belenos} (QPU B). The
experimental setups consist in multi-photon states
produced by a time-demultiplexed semiconductor quantum dot based on demand
single-photon source, evolved through a twelve- or twenty-four-mode photonic
circuit made of an array of fully-programmable Mach--Zehnder interferometers
implementing Clements \emph{et. al} universal design
\cite{clements_optimal_2016}. Photons are detected at the output of the chip by
superconducting nanowire single-photon detectors (SNSPDs). 

For both platforms, the protocol starts with a state preparation stage,
involving single-photons sources, linear-optical evolution and heralding to
prepare an initial state $\rho$, followed by random linear-optical
interferometer with photon-number measurements and a classical data processing
stage (see \cref{fig:photonicShadows}). In the experiments, we achieve
photon-number resolving measurements via pseudo-PNR data processing (see
\arxiv{\cref{sec:ppnr}}{Supplementary Information}). Chip reconfiguration is costly, yet it remains necessary during the data-collection phase. Moreover, because of the number of photons used in the experiment, a large number of samples can be collected. For these reasons, we adopt a multi-shot protocol, in which several samples are taken for the same unitary. The experimental setups are detailed further in the Methods. 

We demonstrate the versatility of our photonic classical shadows protocol with
five applications: evaluating low-order correlation functions
(\cref{exp:loc}), evaluating Lie-algebraic linear optical invariants
(\cref{exp:inv}), estimating binned-probability distributions
(\cref{exp:bpb}),  estimating the energy of the Bose-Hubbard Hamiltonian
(\cref{exp:gee}), and learning Boson Sampling states (\cref{exp:bsl}). For
each experiment, we briefly explain the theoretical background behind it and
report the experimental results in
\cref{fig:experimentsAscella,fig:experimentBelenos} respectively. Finally, we
discuss the scaling of the approaches. In all cases, we emphasise that the
experimental protocols have the same structure---state preparation, followed by
our photonic shadow protocol---while differing at the classical post-processing
stage.

\subsection{QPU A: \emph{Ascella}} \label{sec::qpuA} We report in
\cref{fig:experimentsAscella} the results of the experimental implementation
of the classical shadow protocol for \cref{exp:loc,exp:inv,exp:bpb}. We
employ a twelve-mode integrated universal interferometer, and we utilize the unused modes to
genuinely emulate photon-number resolving measurements with threshold detectors
on chip.   
The input state is obtained from a random preselected evolution $U_{\text{prep}}$ applied to $\ket{1,1,1,0}$.
Once $U_{\text{prep}}$ is fixed, to proceed to the
data-collection phase, we further sample $N=1100$ random unitary evolutions $U_i$. Crucially, we demonstrate the resourcefulness of the input
state's shadow by computing all the results from a unique classical shadow,
i.e., the data collection phase is ran once, and the classical shadow thus
constructed is used in all three applications. 

\medskip

\begin{application}[Low-order correlators]\label{exp:loc} Low-order quantum
correlators, whose purpose are to quantify the correlations of detection events,
comprise a large amount of information about many-body states. They are at the
core of numbers of applications, including bosonic quantum walks
\cite{mayer_counting_2011}, the study of statistical properties of many-particle
interferences for certifying Boson Sampling experiments
\cite{walschaers_statistical_2016,giordani_experimental_2018,flamini_validating_2020},
or the measurement of many-particle entanglement \cite{brunner_manybody_2022}.
One of the simplest mode correlators is the two-mode correlator
$
    C_{i,j} = \expval*{\hat n_i\hat n_j} - \expval*{\hat n_i}\!\!\expval*{\hat n_j},
$
where $\hat n_k=\hat a_k^\dagger \hat a_k$ is the \text{$k$-th} mode number
operator. Higher-order correlators can be obtained from higher-degree products
of the mode number operators straightforwardly.

For the two-mode correlators, we report in \cref{fig:expTPC} an estimation of
the correlation matrix $(C_{i,j})_{1\leq i, j, \leq m}$ with an entry-wise
average absolute error of $0.03$.
\end{application}

\medskip

\begin{application}[Lie-algebraic linear optical invariants]\label{exp:inv}
Linear optical invariants are quantities (or, properties of a state) that remain
constant under linear optical evolution
\cite{migdal_multiphoton_2014,parellada_lie_2026,draux_invariants_2025}. On the
theoretical side, such invariants can be used to derive no-go theorems for Fock
state to Fock state transformations \cite{parellada_nogo_2023}. Moreover, the
experimental observation of Lie-algebraic invariant, as achieved in
\cite{rodari_observation_2025,yang_unveiling_2025}, plays a crucial role for
the characterization of integrated photonic processors. For instance, they allow
for the detection of ongoing nonlinear effects. We demonstrate here that the
classical shadow protocol can be used to measure linear-optical invariants.

Let $\{O_i\}_{i=1}^{m^2}$ be a basis of the space of linear optical Hamiltonians
over $m$ modes and $\rho$ denote an arbitrary $m$-mode Fock state (details are
given in \cref{sec:BoseHubbard}). We experimentally evaluate three
linear-optical invariants. The first (and arguably the simplest) is $I(\rho) =
\sum_{i=1}^{m^2}\expval{O_i}^2.$ The second is the spectrum of $\rho_T =
\sum_{i=1}^{m^2} \expval{O_i} O_i$, which we write $\lambda(\rho_T)$. The third
is the spectrum of the covariance matrix $\Gamma(\rho)$ with entries
$
    \lpr\Gamma(\rho)\rpr_{i,j} 
        = \expval*{O_i}\!\!\!\expval*{O_j} - \frac{1}{2}\expval*{O_iO_j - O_jO_i},
$
which we write $\lambda(\Gamma(\rho))$. Measuring the latter directly is
challenging experimentally as it requires measuring degree-two observables,
while our classical shadow protocol efficiently achieves this task.

We successfully demonstrate in \cref{fig:expINV} that their values can be
recovered efficiently. The experimental results are compatible with the
theoretical results, despite a slightly larger deviation than Ref.\
\cite{rodari_observation_2025}. We obtain $I(\rho) = 2.98$ compared to a
theoretical value of $I(\ket{1110}) = 3$ giving a relative error of $0.67
\%$. 
\end{application}

\medskip

\begin{application}[Binned probabilities]\label{exp:bpb} Binned distributions,
i.e., output distribution of photons distributing in partitions of the output
modes, are yet another tool to validate the correct functioning of linear
optical devices \cite{seron_efficient_2024} which is experimentally feasible
\cite{anguita_experimental_2025}. Probing the total variation distance between
an actual experiment binned-distribution and that of the expected distribution
allows to one to detect fallacious boson samplers.

The $m$ output modes are clustered into $K$ bins ${\Kc = \{\Kc_1, \cdots,
\Kc_K\}}$,
and binned-probabilities are given by the multidimensional discrete Fourier
transform of the characteristic function of bin-number operators.
Importantly, binned-distribution can be approximately sampled from classically
within small total variation distance if the number of bins $K$ is constant
\cite{seron_efficient_2024}. In this setting, the computation of the expected
distribution can be performed classically and used to validate the experimental
setup at hand.

We show in \cref{fig:expBPB} that the binned-distribution can be recovered within an average TVD of
$0.020$ across all bipartitions of the modes, and an average TVD of $0.031$ over
all two-mode binning which is of the order of the experimental results reported
in Ref.\ \cite{anguita_experimental_2025}.
\end{application}

\subsection{QPU B: \emph{Belenos}}

We report in \cref{fig:experimentBelenos} the experimental results obtained
for \cref{exp:gee,exp:bsl} on \emph{Belenos}.

\medskip

\begin{application}[Hamiltonian energy estimation]\label{exp:gee} The
    Bose-Hubbard Hamiltonian is parametrized by a hopping strength ${J>0}$ and interaction strength
    $U$ and chemical potential $\mu$, that describes interacting bosonic
    particles in a lattice.
    In the
    superfluid regime ($U = 0$), the ground state can be produced from a Fock
    basis state acted on by linear optical transformation
    \cite{yalouz_encoding_2021}. 
    
    For the experiment, we consider a hopping strength $J = 1$, and
    set $\mu = 0$ and $U=0$. We use a classical shadow of the ground state that
    we prepare in order to estimate its ground energy. The state preparation
    step consists in producing the state $\ket{3}$ (resp. $\ket{4}$) which we achieve by sending three (resp. four) single photons
    in different modes, evolving them in a
    Fourier interferometer (which generalises the Hong--Ou--Mandel effect) and
    post-selecting on measuring the vacuum in all but the last mode. 
    Our experimental results demonstrate that classical shadows of small size
    suffice to estimate the ground energy within $5\%$. 
    Additionally, we numerically investigate the interacting regime $(U > 0)$
    where the ground state cannot be prepared by linear optical quantum dynamics
    alone in \arxiv{\cref{sec:BoseHubbard}}{the Supplementary Information}. We show that the classical shadow
    protocol successfully estimate the ground energy of the Bose-Hubbard
    Hamiltonian in this regime.
\end{application}

\medskip

\begin{application}[Learning Boson sampling states]\label{exp:bsl} Fock basis
    states acted on by passive linear optical transformations are at the core of
    many applications of near-term photonic quantum computing, including quantum
    simulation \cite{sturges_quantum_2021} or photonic quantum machine learning
    \cite{gan_fock_2022,yin_experimental_2025,salavrakos_errormitigated_2025},
    and sampling tasks conjectured to be classically hard such as Boson Sampling
    \cite{aaronson_computational_2011}. We refer to such states as Boson
    Sampling states. Interestingly, it was recently shown that such states can
    be learned efficiently \cite{iosue_higher_2025}. More precisely, given
    access to copies of $\ket{\psi} = \varphi_m(W)\ket{\bm n}$ for an unknown $W
    \in \U(m)$ and input occupancy $\bm n$, it is possible to recover a $\bm t$
    and construct a $V$ satisfying $\|V - W\Phi P\| \leq \varepsilon$, for some
    irrelevant diagonal unitary matrix $\Phi$ and a permutation matrix $P$, such
    that $\abs{\mel{\bm n}{\varphi_m^\dagger(W)\varphi_m(V\Phi P)}{\bm t}} \geq 1-\delta$,
    with some $\delta$ dependent on $\varepsilon$. Importantly, $\bm t$ and $V$
    can be learned from overlaps with non-Hermitian operators. Luckily,
    our classical shadow protocol works for estimating the overlap with any
    linear operator and not only quantum observables.
    We provide the first experimental implementation of a Boson Sampling state
    learning task via our classical shadow protocol, following the strategy
    outlined in \cite{iosue_higher_2025}. We prepare the input state $\ket{\bm
    n} = \ket{1,1,1,1}$ and use $W = U_{\text{prep}}$ as in \cref{sec::qpuA}
    as unknown unitary evolution. Then we proceed to the data-collection phase. Here, we demonstrate that the input Fock state
    is exactly recovered, and we learn a unitary matrix $V$ such that $\|V - W\|
    = 0.12$, up to the irrelevant phases and permutation. We find that the fidelity between $\varphi_m(W)\ket{\bm n}$ and $\varphi_m(V)\ket{\bm t}$
    is 0.96.
\end{application}

\section{Discussions}

In this work, we have introduced a classical shadow protocol tailored to passive
linear optics with photon-number resolving (PNR) measurements and provided an
extensive experimental demonstration of the scheme. We demonstrated that meaningful and
scalable property learning is possible, thereby extending the realm of classical
shadows to experimentally and publicly accessible photonic platforms. 

We have shown that the scheme is highly versatile as demonstrated experimentally
in five types of applications. In all these applications, we found that
classical shadows provides a scheme for  predicting properties of photonic
quantum states up to small deviations. Several factors that we identify explain
these small deviations. The first factor is the total variation distance of the
sampled distribution with respect to the true distribution at the
data-collection level due to pseudo-PNR. We explain in
\arxiv{\cref{sec:ppnr}}{the Supplementary Information} how this can be
mitigated by taking more samples. Second, the deviations originate from typical
experimental imperfections of the apparatus, whose main representative are
partial photon distinguishability, chip characterization, and multi-photon
emission (see \cref{app:exp-meth} for details). 

The input state generation rate arguably constitutes the primary limitation to
experimental realizations. The reason is that, as of today, spatially encoded
photonic states are produced from a single time-demultiplexed single-photon
source. Considerable research efforts are directed towards using multiple
single-photon sources  
\cite{pont_indistinguishability_2025}, as this is an important building-block of
photonic fault-tolerant architectures. Along with experimental demonstrations of
measurement-based state preparation
\cite{prevedel_highspeed_2007,thiele_cryogenic_2025}, this suggests that
highly-complex states will shortly be within reach. 

On the theoretical front, while we have focused on estimation of linear properties, we expect our protocol to remain efficient for nonlinear functions of quantum states of constant degree in the creation and annihilation operators via U-statistics \cite{huang_predicting_2020}.

Beyond predicting physical properties, another promising venue for classical
shadows obtained from a quantum computer are machine learning applications
\cite{huang_power_2021,jerbi_shadows_2024,cho_machine_2024}. For instance,
classical shadows can be used to construct a classical representation of the
output of a parametrized circuit that is then employed during a fully classical
classification phase. Photonic variational quantum circuits have recently
attracted a lot of attention, both theoretically
\cite{gan_fock_2022,pappalardo_photonic_2025} and experimentally
\cite{salavrakos_errormitigated_2025,hoch_quantum_2025,yin_experimental_2025}. 
We anticipate that our classical shadow protocol will help scale up these experiments thanks to its ease of use.

Additionally, it will be crucial to study the effect of noise on the protocol at
larger scale. Photonic architectures are subject to characteristic types of
noise, namely, photon losses and photon indistinguishability. On this basis,
incorporating error mitigation techniques tailored to linear optics
\cite{mills_mitigating_2026,taylor_quantum_2024} can be considered.

The development of integrated photonic technologies enables the fabrication of
large-scale and publicly available machines. On the other hand, the framework of classical shadows
is a powerful tool that benefits from an ever-growing number of refinements \cite{jnane_quantum_2024,hu_demonstration_2025,chen_robust_2021,west_real_2025,koh_classical_2022,zhao_fermionic_2021,wan_matchgate_2023,sauvage_classical_2024,ippoliti_classical_2024,low_classical_2024,gandhari_precision_2024,becker_classical_2024,nguyen_optimizing_2022,helsen_thrifty_2023,zhou_performance_2023,chen_adaptivity_2024,fawzi_learning_2024}. In
this context, our work naturally opens the door to a wealth of new applications
for photonic quantum computers via classical shadows \cite{hadfield_measurements_2022,mcnulty_estimating_2023,dutt_practical_2023,huang_provably_2022,huang_power_2021,jerbi_shadows_2024,haug_quantum_2023}.

\medskip

\begin{acknowledgments}
\noindent The authors acknowledge Hela Mhiri, Emilio Annoni, Ariane Soret,
Leonardo Banchi, Micha\l{} Oszmaniec, Markus Heinrich, Tommaso Francalanci,
Fabio Sciarrino and Elham Kashefi for fruitful discussions. This work has been
co-funded by the European Commission as part of the EIC accelerator program
under the grant agreement 190188855 for SEPOQC project, by the Horizon-CL4
program under the grant agreement 101135288 for EPIQUE project, the TUF-TOPIQC
project part of the Trilateral Call for Quantum Innovation, co-financed by
Germany, the Netherlands and by the French National Quantum Strategy (France
2030) program, and by the CIFRE grant n$\degree$2023/1746.
\end{acknowledgments}

\newpage

\bibliography{bibliography,urls}

\onecolumngrid

\appendix

\newpage

\arxiv{\newpage}{}

\begin{center}
 {\Large \textbf{Appendix} }   
\end{center}

\noindent \large \textbf{Table of Contents}

\normalsize
\noindent\makebox[\linewidth]{\rule{\linewidth}{0.4pt}}
\begingroup
\makeatletter
\par 
\renewcommand{\contentsname}{Appendix contents}
\@starttoc{atoc}
\makeatother
\endgroup

\noindent\makebox[\linewidth]{\rule{\linewidth}{0.4pt}}

\setcounter{figure}{0}

\section{Experimental methods}
\label{app:exp-meth}

\subsection*{\cref{exp:loc,exp:inv,exp:bpb}}

\paragraph{Hardware.} \emph{Ascella} is a twelve-mode six-photon photonic
quantum processing unit. The detailed specifications can be found in
\cite{maring_versatile_2024}. It is controlled by the Perceval Python library
\cite{heurtel_perceval_2023} and we have developed an open-source Julia package provided here
\cite{github_loshadows} for photonic classical shadows, designed to work seamlessly with Perceval.

\paragraph{State preparation.} The input state for the experiment is obtained by
preparing the 3-photon 4-mode Fock state $\ket{1,1,1,0}$ and letting it evolve
through a fixed state preparation circuit $U_{\text{prep}}$ (chosen at random).

\paragraph{Data-collection phase.}
For each choice of a random matrix $U_i$, we perform 2000 shots (a \emph{shot} is
defined as an event with at least one click), and perform pseudo-PNR measurement
of the output state, yielding an average of $19$ 3-photon events. As pseudo-PNR
induces a bias in the resulting distribution, we perform the mitigation step
described in \arxiv{\cref{sec:ppnr}}{the Supplementary Information}. This requires a sample overhead to achieve perfect
mitigation, which we illustrate in \cref{fig:TVDPPNRExperiment} for measuring
the output distribution of $U_{\text{prep}}$. We take advantage of this to
perform multi-shot shadow \cite{zhou_performance_2023} at a low cost. 

\paragraph{Classical post-processing phase.}
The moderate dimension of the Hilbert space allows us to compute the density
matrix of each snapshot by averaging over all output samples. More precisely,
$\dim \Hc_4^3 = 20$, i.e., each snapshot requires one to compute four-hundred
$3\times 3$ matrix permanents using Ryser's algorithm
\cite{ryser_combinatorial_1973}. Therefore, the classical shadow we obtain is a
set of 1100 density matrices. 

\paragraph{Property prediction phase}
Due to the small size of the classical shadow $\S(\rho, 1100)$, expectation
values of observables are obtained by computing the average
\begin{equation}
    \tr{O\rho} \approx \frac{1}{1100}\sum_{\hat \rho \in \S(\rho, 1100)} \tr{O\hat\rho}.  
\end{equation} 
Storing the matrix representation of the observables and performing the matrix
multiplication $O\hat\rho$ is not an issue for such an experiment, but this begs
the question of having an efficient representation of quantum observables when
both the number of modes and photons grow. We address this issue for \cref{exp:gee,exp:bsl} by showing that expectation values of observables can be obtained from the polynomial representation of the observables.

\begin{figure*}[!t]
    \centering

    \subfloat[ \justifying Configuration of the chip for experiment reported in
        \cref{fig:experimentsAscella}. The Pseudo-PNR phase consists in plugging
        each output mode of the evolution step to the top mode of a 3-mode
        Fourier interferometer ($\Fc_3$)---maximising the scattering probability
        of the input state into different modes---whose output is measured with
        thresholds SNSPD to yield $\bm s = (s_1, s_2, s_3, s_4)$ (see
        \cref{sec:ppnr} for details).\label{fig:ascella}]{
        \includegraphics[width=0.45\linewidth]{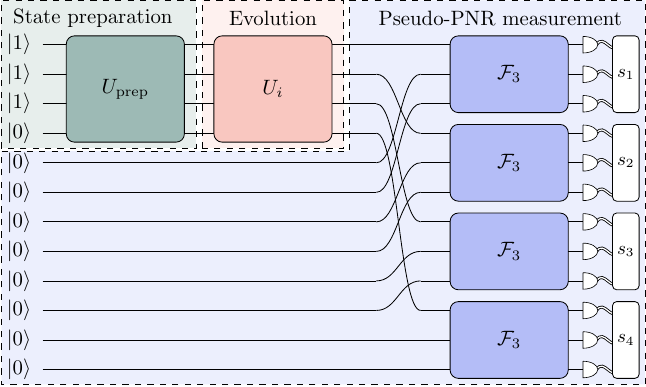}
    }
    \hfill
    \subfloat[ \justifying Effect of the Pseudo-PNR mitigation on the total
        variation distance (TVD) to the distribution obtained with perfect PNR
        measurement for the output distribution associated with the four-mode unitary matrix $U_{\text{prep}}$. The TVD of
        non-mitigated distribution reaches a plateau at around $0.12$. On the
        contrary, the TVD of the mitigated distribution eventually attains zero
        provided enough samples are available.\label{fig:TVDPPNRExperiment}]{
        \includegraphics[width=0.45\linewidth]{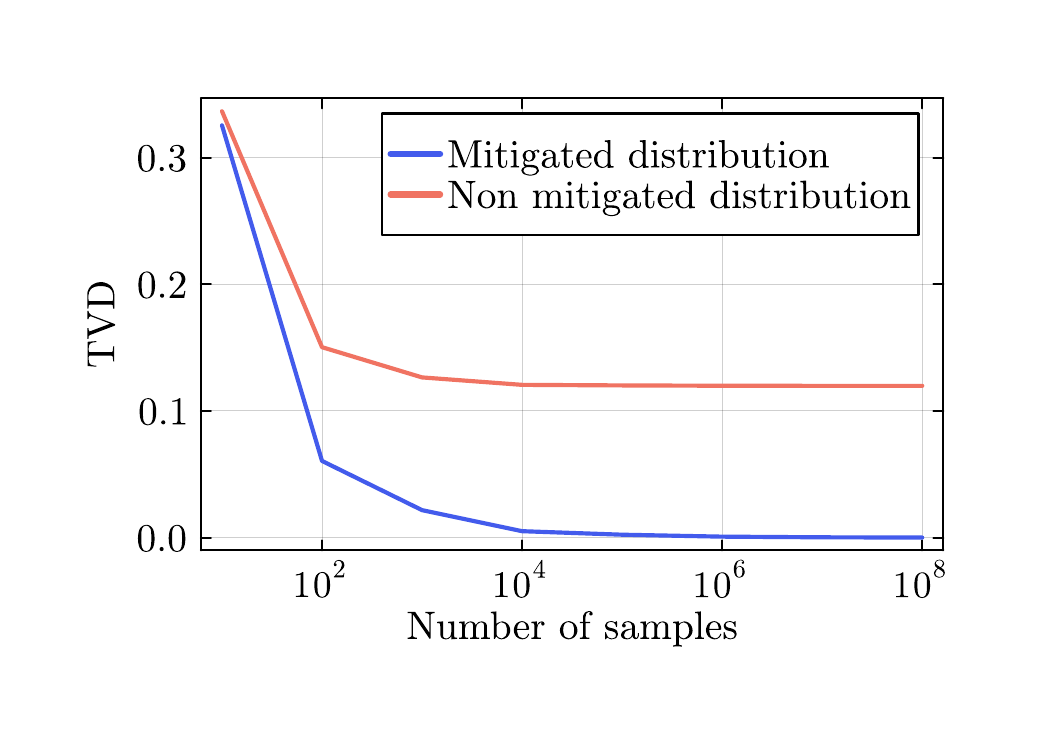}
    }

    \caption{
        \textbf{Pseudo photon-number resolving measurement.}
        Illustration of the pseudo-PNR circuit and the effect of mitigating
        the output distribution on the total variation distance.
    }
    \label{fig:ppnr}
\end{figure*}

\paragraph{Matrix-form of the channel.}
A crucial step for the photonic shadow protocol is the evaluation of the quantum
channel $\Ms^{-1}$. Quantum channels acting on $d$-dimensional operators are
represented by $d^2\times d^2$ matrices, and applying the channel reduces to
performing the matrix-vector multiplication, the vector being the vectorised
form of that said operator. Therefore, in matrix form, $\Ms^{(n)^{-1}}$ is
described by a $\dim\Hc_m^{n^2}\times \dim\Hc_m^{n^2}$ matrix, which, even for
small-scale experiment, is unmanageable on a modern personal computer.

As described in \arxiv{\cref{sec:closedChannelForm}}{the Supplementary Information}, the channel matrix is populated with
sums of Clebsch-Gordan coefficients. We observe that the scaling of the number
of nonzero Clebsch-Gordan coefficients is moderated compared to $\dim
\Hc_m^{n^4}$, which makes the projectors $\Pi_{\lambda_k}$ sparse
symmetric matrices of order $\dim \Hc_m^{n^2}$. The upper triangular part is
stored in memory in Compressed Sparse Column (CSC) format, and the matrix-vector
multiplication is performed using
\begin{equation}
  Ax = A^{\triangle}x + (x^TA^{\triangle})^T - \mdiag(A)x,
\end{equation}
where $A^{\triangle}$ is the upper-triangular part of $A$ and $\mdiag(A)$ its
diagonal. We provide a Julia implementation of the computation of the matrix
form of the channel\arxiv{ (see in particular \cref{eq:channelMatrixForm})}{}, together
with an interface for using the Perceval \cite{heurtel_perceval_2023} (Python)
library. 

The calculation of the matrix-form of the channel is a computationally demanding
task; we show in the \cref{app:exactTechnique} how the explicit calculation can
be avoided. We use this refined procedure to avoid the computation of the
matrix-form of the channel for \cref{exp:gee,exp:bsl}.

\subsection*{\cref{exp:gee,exp:bsl}}

\paragraph{Hardware.} \emph{Belenos} is a
twenty-four-mode cloud accessible photonic quantum processing unit
\cite{quandela_cloud}.
Temporal-to-spatial multiplexing is implemented using a resonantly driven Pockels-cell-based demultiplexer (overall transmission exceeding 85\%). The system includes custom polarization control modules to prepare and stabilize the input states. The photonic integrated circuit used for light manipulation is a 24-mode silica glass interferometer \cite{barzaghi_24mode_2025} with 552 beam splitters integrated as directional couplers and 676 reconfigurable thermo-optic phase shifters driven by custom low-noise current electronics. The circuit is calibrated using machine learning \cite{fyrillas_scalable_2024}. The processor implements a self-calibration strategy, which automatically recalibrates the photonic integrated circuit using periodically acquired data samples, to maintain a high-accuracy of control over the long term. The detection system consists of 16 polarization-resolved photon-number-resolving (P-PNR) detectors, each based on 28 interleaved pixels, enabling reliable discrimination of up to two simultaneous photons per spatial mode via multi-trigger events and 8 superconducting nanowire single-photon detectors (SNSPD). The single-photon detection probability is measured to be 90–92\%, while the two-photon detection probability reaches $\sim$78\%.
Similar to \emph{Ascella}, it is controlled by the Perceval Python library \cite{heurtel_perceval_2023}.

\paragraph{State preparation for \cref{exp:gee}.} The ground state
preparation of the Bose-Hubbard Hamiltonian can be done by acting with a passive
linear optical circuit on a Fock basis state. In practice, it consists in
preparing the resource input state $\ket{n}$ from $n$ single-photons in
different modes, which is achieved using a quantum Fourier interferometer and
post-selecting on measuring $0$ photons in the first $n-1$ modes. The state thus
obtained is fed into a cascade of parametrized beam-splitters that prepare the
ground state. We show the overall circuit, including the measurement step, in
\arxiv{\cref{fig:BelenosForBH}}{the Supplementary Information}. To showcase the capability of our protocol to learn properties of states beyond linear-optics, we numerically try and learn the energy of the ground state of the Bose-Hubbard Hamiltonian in the regime $U>0$, that can provably not be implemented with linear-optics only \cite{yalouz_encoding_2021}.

\paragraph{Classical post-processing phase.}
The classical post-processing phase for \cref{exp:gee,exp:bsl} does not rely
on computing the matrix form of the channel to find the projections. Rather, the
decomposition in each subspace is obtained via a procedure described in
\arxiv{\cref{pr:individualProj}}{the Supplementary Information}. Moreover, we also show \arxiv{in \cref{pr:degreeKObs}}{} that expectation values of low-degree observables
can then be efficiently computed. The algorithm we introduce for exactly computing expectation values of observables of degree $d$ exploits the fact that if $d = O(1)$, they can be described by a polynomial number of $d \times d$ matrix permanents, each of which can be computed efficiently exactly with Ryser's algorithm \cite{ryser_combinatorial_1973}.

\arxiv{
\begin{figure*}[!ht]
    \includegraphics{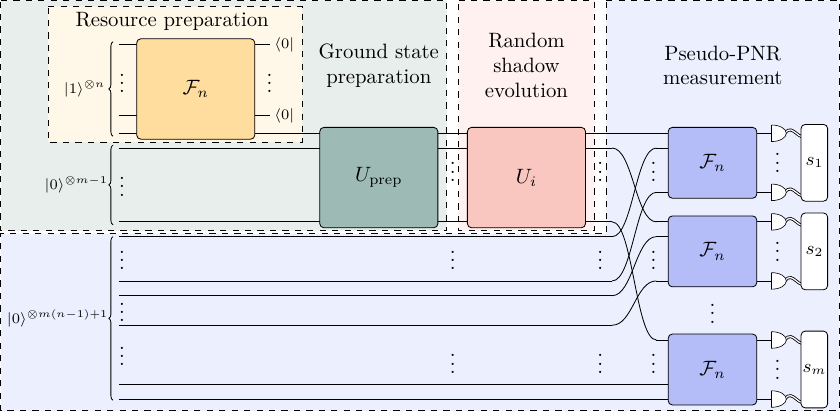}
    \caption{\justifying \textbf{Illustration of the circuit implemented on
    Belenos}. The resource preparation step consists in preparing the Fock state
    $\ket{n}$ from $n$ single-photons via a QFT interferometer and
    post-selecting on the first $n-1$ modes being occupied by the vacuum.
    $U_{\text{prep}}$ implements the ground state via a cascade of
    beam-splitters. The remainder of
    the circuit, which consists of a random evolution followed by pseudo-PNR
    measurement, mimics that the experimental runs on Ascella.}
    \label{fig:BelenosForBH}
\end{figure*}}{}

\section{Background}

We first introduce the notations we use hereafter, then recall the classical
shadow protocol and provide the necessary tools of representation theory in
linear optics.

\subsection{Notations}

\paragraph{Linear algebra.}
Let $\Hc$ be a complex Hilbert space and $B$ and orthonormal basis of $\Hc$. We
denote by $\Bs(\Hc)$ the set of bounded operators on $\Hc$, and by
$\Ls(\Hc)$ the set of linear operators on $\Hc$. $\Ls(\Hc)$ is itself a
Hilbert space equipped with the inner product $\braket{X}{Y} = \tr{X^\dagger
Y}$. The vector space $\Ls(\Ls(\Hc))$ is the space of so-called
\emph{superoperators} on $\Hc$. Given a linear operator $X \in \Ls(\Hc)$, we
define its \emph{vectorised} form by the operator 
\begin{equation}
    \ket{X} = \sum_{\x, \y \in B} \mel*{\x}{X}{\y} \ket*{\x, \y}
\end{equation}
carried by the tensor product space $\Hc \otimes \Hc$. This operation is indeed
compatible with the inner product over $\Ls(\Ls(\Hc)) \simeq \Ls(\Hc \otimes
\Hc)$.

\paragraph{Multi-indices.}
We use bold fonts to denote vectors $\s = (s_1, \cdots, s_m) \in \N^m$ for $m\in
N^*$. We define the following operations on multi-indices:
\begin{equation}
    \begin{aligned}
        |\s| & = s_1 + \cdots + s_m, \\
        \s! & = s_1!\cdots s_m!, \\
        \s \pm \t & = (s_1 \pm t_1, \cdots, s_m \pm t_m), \\
        \ket{\s} & = \ket{s_1, \cdots, s_m}. \\
    \end{aligned}
\end{equation}
Finally, we denote the set of all $m$-tuples of natural numbers summing to $n$
\begin{equation}
    \Phi_{m}^n = \{\s\ | \ |\s| = n\} \subset \N^m.
\end{equation}

\subsection{Classical shadows}
In this section, we briefly recall the classical shadow formalism. For a more
thorough introduction we refer the reader to
\cite{huang_predicting_2020,elben_randomized_2023}. We consider an unknown
$n$-qubit quantum state. The aim is to estimate expectation values $o_t =
\Tr{O_t\rho}$ for a set of $T$ observables $\{O_t\}_{t=1}^T$. The protocol
consists in two steps: 1) the data acquisition phase and 2) a classical
post-processing. Given a set of unitary matrices $\Uc$ and a probability measure
$\mu$ over $\Uc$, the copies of $\rho$ are evolved by drawing at random $U \in
\Uc$ from $\mu$, and measuring in a basis $\Bc$ by a set of projectors
$\{\Pi_b\}_{b \in \Bc}$. This procedure can be seen to implement the
quantum channel 
\begin{equation}
    \Mc(\rho) = 
        \sum_{b \in \Bc} 
        \intg{\Uc}{\tr{U\rho U^\dagger \ketbra{b}} U^\dagger \ketbra{b}U}{\mu(U)}.
\end{equation}
For each choice of pair $(U, b)$, the \emph{snapshot} of $\rho$ is defined
as 
\begin{equation}
    \hat \rho_{(U, b)} = \Mc^{-1}(U^\dagger \ketbra{b}U).
\end{equation}
In the rest, we may omit the subscript and simply denote a snapshot $\hat\rho$.
The inverse of $\Mc$ might not be physical but exists and is unique if the
ensemble of unitary transformations defines a tomographically complete set of
measurements \cite{huang_predicting_2020}. If the ensemble of unitary is not
tomographically complete, $\Mc^{-1}$ denotes the Moore-Penrose pseudo-inverse of
$\Mc$. The collection of all snapshot one obtains is referred to as a
\emph{classical shadow} of $\rho$. In practice, it will be more convenient to
store in memory the pair $(U, b)$ rather than $\hat \rho$. A closed form for
$\Mc$ follows from representation theory of the unitary group (and sub-groups
thereof) and Weingarten calculus \cite{kostenberger_weingarten_2021} --- see
\cite[Sup. Mat.]{garcia-martin_quantum_2025} for a detailed introduction. For
instance, $\Mc$ is a depolarising channel for shadows based on random Clifford
measurements \cite{huang_predicting_2020} and other choices of $\Uc$ induce
other measurement channels, see e.g.,
\cite{west_real_2025,sauvage_classical_2024}.

Remarkably, to estimate the expectation value of a set of $T$ observables to
within additive error $\varepsilon$ with constant failure probability, one only
needs
\begin{equation}\label{eq:sampleComplexity}
    N = \bigo{\frac{\log(T)}{\varepsilon^2} \max_t \mathrm{Var}(\hat o_t)}
\end{equation}
copies of the unknown state $\rho$, where $\hat o_t = \Tr{O_t \hat\rho}$ is the
classical shadow estimate of $o_t$.
Beyond expectation values of observables, nonlinear functions such as Rényi
entropies or fidelities can also be efficiently estimated within this framework.
The image of $\Mc$ is called in \cite{kirk_hardwareefficient_2022} the
\emph{visible space}. More formally, it is defined as
\begin{equation}
    \VS(\Uc, \Bc) = \text{span}\aco U^\dagger \ketbra{b} U\acf_{U \in \Uc, b \in \Bc}.
\end{equation}
Given an observable $O$, the authors show that $\hat o$ is unbiased if ${O \in
\VS(\Uc, \Bc)}$. Otherwise, the observable can be written as $O = O_v + O_{\bar
v}$, where $O_v \in \VS(\Uc, \Bc)$ and $O_{\bar v} \notin \VS(\Uc, \Bc)$. A pair
$(\Uc, \Bc)$ is said to be \emph{tomographically complete} if and only if
$\Ls(\Hc) = \VS(\Uc, \Bc)$. We use this formalism to quantify the amount of
information that can be recovered with the protocol we introduce.

\subsection{Representation theory of linear optics}

We hereafter denote by $m\in \N$ the number of modes and $n\in \N$ the number of
photons. We refer to
\cite{banchi_multiphoton_2018,wilkens_benchmarking_2024,arienzo_bosonic_2025}
for a more in-depth exposition of representation theory for linear optics. The
creation and annihilation $\hat a_k$ and $\hat a_k^\dagger$ for all $1 \leq k
\leq m$ satisfy the canonical commutation relations
\begin{equation}
    [\hat a_j,\hat a_k^\dagger] = \delta_{j,k}\mathbb I, \qquad [\hat a_j,\hat a_k] = [\hat a_j^\dagger, \hat a_k^\dagger] = 0.
\end{equation}
They are described by infinite-dimensional matrices whose entries are
$\mel{m}{\hat a}{n}$ and $\mel{m}{\hat a^\dagger}{n}$ respectively. The carrier
Hilbert space of this operator algebra is the infinite-dimensional $m$-mode
Fock-Hilbert space 
\begin{equation}
    \Fc_m = \bigoplus_{n \geq 0} \Hc_m^n,
\end{equation}
where $\Hc_m^n$ is the subspace of $n$ photons distributed over $m$ modes
endowed with the orthonormal basis $\{\ket{\n}\}_{\n \in \Phi_m^n}$. In
particular, an $n$-particle basis Fock state over $m$ modes is of the form
\begin{equation}
    \ket{\n} 
    = \ket{n_1, \cdots, n_m} 
    = \prod_{i=1}^m\frac{(\hat a_i^\dagger)^{n_i}}{\sqrt{n_i!}}\ket{\text{0}}^{\otimes m}.
\end{equation}

The set of possible $m$-mode passive transformations is identified with the
unitary group $\U(m)$. A matrix $U \in \U(m)$ acts on the vector of creation
operators as 
\begin{equation}\label{eq:evolutionCreationOperators}
    \begin{pmatrix}
        \hat a_1^\dagger\\
            \vdots\\
        \hat a_m^\dagger
    \end{pmatrix}
    \mapsto
    U \begin{pmatrix}
        \hat a_1^\dagger\\
            \vdots\\
        \hat a_m^\dagger
    \end{pmatrix}
    = 
    \begin{pmatrix}
        \sum_{k=1}^mu_{1,k}\hat a_k^\dagger \\
        \vdots\\
        \sum_{k=1}^mu_{m,k}\hat a_k^\dagger \\
    \end{pmatrix}.
\end{equation}

The building blocks of the unitary transformations are phase-shifters and beam
splitters, and a universal interferometer implementing $m\times m$ unitary
transformations can be designed by arranging these building blocks in an efficient
rectangular \cite{clements_optimal_2016} or triangular
\cite{reck_experimental_1994} layout. This implies that the linear-optical
network implementing the transformation $U$ can be found by algorithmic means.

The action of $U$ on the multimode Fock basis is given by the mapping
\begin{equation}
    \varphi_m: \U(m) \ni U \mapsto \varphi_m(U) \in \U(\Fc_m),
\end{equation}
where $\U(\Fc_m)$ is the group of all unitary transformations on $\Fc_m$. The
map $\varphi_m$ is a unitary representation of $\U(m)$ on $\Fc_m$. It can be
decomposed as an orthogonal sum of finite-dimensional sub-representations as
\begin{equation}\label{eq:decompositionPhiM}
    \varphi_m = \bigoplus_{n \geq 0} \varphi_m^n,
\end{equation}
where $\varphi_m^n$ acts irreducibly on the $n$-photon sector
\cite{aniello_exploring_2006}. Using this decomposition, $\varphi_m^n(U)$ is a
$|\Phi_m^n|\times |\Phi_m^n|$ matrix whose entries are 
\begin{equation}\label{app:eq:bsamplitudes}
    \mel{\s}{\varphi_m^n(U)}{\t} = \frac{\per{U_{\s,\t}}}{\sqrt{\s!\t!}},
\end{equation}
where $\per{\cdot}$ is the matrix-permanent function
\cite{marcus_permanents_1965} defined as 
\begin{equation}
    \per{A} = \sum_{\sigma \in S_n} \prod_{i=1}^n a_{i, \sigma(i)}
\end{equation}
for an $n \times n$ matrix $A = (a_{i,j})_{1 \leq i,j\leq n}$, and $U_{\s, \t}$
is a matrix obtained from $U$ by copying $s_i$ times its $i$-th row and $t_j$
times its $j$-th column \cite{aaronson_computational_2011}.

The space of operators acting on a Hilbert space $\Hc$ is denoted $\Ls(\Hc)$.
Operating on density matrices, we write the action of $\omega_m$ by
conjugation on $\rho$ as
\begin{equation}
    \omega_{m}(U)(\rho) = \varphi_{m}(U) \rho \varphi_{m}(U)^\dagger.
\end{equation}
Similarly to \cref{eq:decompositionPhiM}, $\omega_m$ admits the
following block-decomposition:
\begin{equation}
    \omega_m = \bigoplus_{n\geq 0} \omega^n_m.
\end{equation}

Although $\varphi^n_m$ is an irreducible representation, $\omega^n_m$ is not. It
is convenient to define the \emph{vectorised} form of $\varphi^n_m$ as the outer
product ${\omega^n_m \simeq \varphi^{n}_{m} \otimes \bar\varphi^{n}_{m} \in
\Ls(\Hc_{m, n} \otimes \Hc_{m, n}^*)}$. Irreducible representations of $\U(m)$ are
in one-to-one correspondence with Young diagrams \cite{fulton_young_1996}.
Diagrammatically, Littlewood-Richardson rules describe how the tensor product of
irreps can be decomposed into a direct sum of irrep
\cite{fulton_young_1996,sternberg_group_2003,georgi_lie_2018}. 

In the context of indistinguishable bosonic particles, we only consider the
totally symmetric irrep of $\U(m)$. In particular, it was shown in
\cite{arienzo_bosonic_2025,wilkens_benchmarking_2024} that by Maschke's theorem,
$\omega^n_m$ can be decomposed into $n+1$ inequivalent multiplicity-free
irreps
\begin{equation}\label{eq:omegaDirectSumLamdas}
    \omega^n_m
        \simeq \bigoplus_{k = 0}^n \lambda_k,
\end{equation}
where $\lambda_0$ is the trivial irrep. The transformation is realised by the
so-called Clebsch-Gordan matrix, the entry of which are called Clebsch-Gordan
coefficients \cite{folland_course_2016,fulton_representation_2004}.
Clebsch-Gordan coefficients are numbers that arise the context of angular
momentum coupling. These coefficients can be computed in polynomial time
\cite{alex_numerical_2011}. Under this change of basis we obtain a similar
decomposition for $\Ls(\Hc_m^n)$, known as its \emph{isotypic decomposition}:
\begin{equation}\label{eq:isotypicDecomposition}
    \Ls(\Hc_m^n) \simeq \bigoplus_{k=0}^n \Hc_{\lambda_k},
\end{equation}
where $\Hc_{\lambda_k}$ is the subspace of $\Ls(\Hc_m^n)$ carrier of the irrep
$\lambda_k$. An orthonormal basis for $\Hc_{\lambda_k}$ is identified by
$\{\ket{\gt{m}}\}_{\gt{m} \in GT(\lambda_k)}$, where $GT(\lambda)$ denotes the
set of Gelfand-Telstin patterns associated with irrep $\lambda$
\cite{arienzo_bosonic_2025,wilkens_benchmarking_2024,alex_numerical_2011}.
Finally, we denote by $\Pi_\lambda$ the projector onto the irreducible subspace
$\Hc_\lambda$.

In our setting, the visible space is 
\begin{equation}\label{eq:LOVisibibleSpace}
    \bigoplus_{n\geq 0} \VS(\U(\Hc_{m}^n), \Phi_m^n),
\end{equation}
where $\U(\Hc_{m}^n)$ is the image of $\U(m)$ via $\varphi_m^n$. It follows
that the protocol is tomographically complete within each photo-number sector of
the Hilbert space.

Importantly, \cref{eq:LOVisibibleSpace} indicates that the input state needs
not be produced by passive linear optical transformations on Fock basis states.
For instance, states produced by measurement-based processes like feedforward,
state injection, all lie in the visible space of the measurement channel.

As a final remark, we detail the decomposition of an operator using monomials of
creation and annihilation, a set of operators $\{\mathfrak{E}_{\bm p, \bm
q}\}_{\bm p, \bm q}$ satisfying the orthonormality relation
$\tr{\mathfrak{E}_{\bm i, \bm j}{E}_{\bm k, \bm \ell}}=\delta_{\bm i, \bm
k}\delta_{\bm j, \bm \ell}$ is required. An example of such basis is given in
\cite{goldberg_covariant_2024} through a phase-space decomposition, or in
\cite{makinen_reconstruction_2019} via polynomials of the ladder operators. We
use this latter decomposition, which takes the form of 
\begin{equation}
    \rho = \sum_{\bm k, \bm l \in \N^m}
        \expval{\mathfrak{E}_{\bm k, \bm \ell}}_\rho E_{\bm k, \bm \ell},
\end{equation}
where
\begin{equation}\label{eq:dualCoefficient}
    \expval*{\mathfrak{E}_{\bm k, \bm \ell}}_\rho 
    = \expval{\prod_{j=1}^m \frac{1}{k_j!\ell_j!}
        \sum_{q_j = -\min(k_j, \ell_j)}^\infty \frac{(-1)^{q_j}(k_j+\ell_j+q_j)!}{(k_j+q_j)!(\ell_j+q_j)!}
        \hat a_j^{\dagger \ell_j + q_j}\hat a_j^{k_j + q_j}
    }_\rho.
\end{equation}

\section{Shadow tomography for linear optics}\label{sec:channel}

In this section, we describe how to perform shadow tomography of an unknown
$m$-mode Fock state $\rho$, i.e., how to produce a state $\hat\rho$ that behaves
like $\rho$ in expected value. The only a priori knowledge we assume is $m$. We
do not fix the number of photons. In particular, we allow for coherent
superposition of different photon-number states, and even state not prepared by
passive linear optical transformations on Fock basis states (e.g., by
measurement-based processes
\cite{chabaud_quantum_2021,monbroussou_quantum_2025}). As we will show, the
visible space in this setting is constrained by the measurements that project
onto a fixed photon-number subspace.

\subsection{Data collection in linear optics}

The data collection process is similar to the classical shadow protocol for
qubits. A \emph{snapshot} is obtained by drawing a unitary matrix at random from
the Haar measure (which can easily be done numerically
\cite{mezzadri_how_2007}), applying its corresponding operation to a fresh copy
of the unknown state $\rho$ by letting it evolve through the linear optical
interferometer described by $U$, and measuring using photon-number resolving
detectors yields an outcome $\s \in \Phi_m^n$, where $n$ is the number of
measured photons.

The average mapping (over the choice of unitary transformation and measurement
outcome, and henceforth photon-numbers) of the input state $\rho$ takes the form
of a quantum channel $\Mc$, which we write
\begin{equation}
    \bigoplus_{n\geq 0} \e{\omega_m^n(U)^\dagger(\ketbra{\s})} = \Mc(\rho).
\end{equation}
A closed-form for $\Mc$ is given in the next section.

\subsection{A closed form for the channel}\label{sec:closedChannelForm} 

As passive unitary transformations preserve the number of photons and
photon-number detection projects onto one of the subspaces of $\Fc_m$, the
measurement channel is not tomographically complete. In particular, it is clear
that applying $\Mc$ yields a diagonal operator, where each block correspond to
each photon-number subspace of $\Ls(\Fc_m)$. Formally, we
\begin{equation}\label{app:eq:channel}
    \Mc = \bigoplus_{n \geq 0}{\Mc^{(n)}} \circ \Pc^{(n)},
\end{equation}
where $\Pc^{(n)}(\rho) = \sum_{\s,\t \in \Phi_m^n}\Pi_\s \rho \Pi_\t \in
\Ls(\Ls(\Hc_m^n))$ is the superoperator performing a projection of $\rho$ onto
the $n$-photon subspace of $\Ls(\Fc_m)$. That is to say, the off-block-diagonal
elements, corresponding to coherence information between different photon-number
states, are erased in the process. We state \cref{thm:tomography}; the rest
of the section is devoted to its proof.

\begin{theorem}[Fock state classical shadows]\label{thm:tomography} Let ${m\geq
    1}$ denote the number of modes. Let $\rho$ be an $m$-mode, $n$-photon Fock
    state. Define the measurement channel
    \begin{equation}\label{eq:channelMatrixForm}
        \mathscr{M}^{(n)} = \sum_{k=0}^n s_{\lambda_k} \Pi_{\lambda_k},
    \end{equation}
    with $s_{\lambda_k} = \frac{m-1}{2k + m - 1}\binom{k+m-2}{k}^{-1}$, and
    where $k$ runs over all $n+1$ isotypic subspaces of $\Ls(\Hc_m^n)$ and
    $\Pi_{k}$ projects onto $\Hc_{\lambda_k}$ (as defined in
    (\cref{eq:isotypicDecomposition})). Then, 
    \begin{equation}
        \rho = \mathbb{E}_{U, \bm{s}}\lbr\mathscr{M}^{(n)^{-1}}\big(\omega_m^n(U)^{\dagger}(\ketbra{\bm{s}}) \big)\rbr,
    \end{equation}
    where the expectation value is taken over the choice of unitary and of
    measurement outcome. 
\end{theorem}
\begin{proof}
    Recall from \cref{app:eq:channel} that due to the photon-number
    measurement the measurement channel takes the form 
    \begin{equation}
        \Mc = \bigoplus_{n \geq 0}\Mc^{(n)},
    \end{equation} 
    where each fixed-photon-number channel is
    \begin{equation}
        \begin{aligned}
            \Mc^{(n)} 
                & = \sum_{\s \in \Phi_m^n}\intg{U(m)}{\tr{\omega_m^{n}(U)(\cdot) \ketbra{\s}} \omega_m^{n\dagger}(U)(\ketbra{\s})}{\mu(U)} \\
                & = \intg{U(m)}{}{\mu(U)}\omega_m^{n}(U) \mathcal{M}  \omega_m^{n\dagger}(U),
        \end{aligned}
    \end{equation}
    where we defined $\mathcal{M} = \sum_{\s \in \Phi_{m,
    n}}\tr{\ketbra{\s}(\cdot)}\ketbra{\s}$. Using the decomposition of
    $\omega_m^n$ into irreducible representations given in
    \cref{eq:omegaDirectSumLamdas}, the closed form for $\Mc^{(n)}$ follows
    from a direct application of Schur's lemma
    \cite{fulton_representation_2004}:
    \begin{equation}
        \Mc^{(n)} = \sum_{k=0}^{n} \frac{\tr{\mathcal{M}\Pi_{\lambda_k}}}{\tr{\Pi_{\lambda_k}}}\Pi_{\lambda_k}.
    \end{equation}
    We use the same notation of \cite{arienzo_bosonic_2025} for the constants, i.e., we set 
    \begin{equation}\label{app:eq:slambdak}
        s_{\lambda_k} = \frac{\tr{\mathcal{M}\Pi_{\lambda_k}}}{\tr{\Pi_{\lambda_k}}} =  \frac{m-1}{2k + m - 1}\binom{k+m-2}{k}^{-1}.
    \end{equation}
\end{proof}

A given irreducible representation $\lambda$ admits an orthonormal basis,
identified by so-called Gelfand-Tsetlin patterns \cite{moshinsky_gelfand_1966},
that is to say,
\begin{equation}\label{eq:IdLambdaGTBasis}
    \Pi_\lambda = \sum_{\gt{m} \in GT(\lambda)} \ketbra{\gt{m}},
\end{equation}
as $\{\ket{\gt{m}}\}_{\gt{m} \in GT(\lambda)}$ forms an orthonormal basis of
$\Hc_\lambda$. As it is convenient to work in the Fock basis, we now show how to
express $\Id_\lambda$ in the Fock basis (of $\Ls(\Ls(\Hc_m^n))$) when $\lambda$
is an isotypic component of $\omega_m^n$, see \cref{eq:omegaDirectSumLamdas}.
Using \cite{arienzo_bosonic_2025}, each such $\ket{\gt{m}}$ can be expressed in
the GT basis of $\varphi_m^n \otimes \bar \varphi_m^n$ using Clebsch-Gordan
coefficients as 
\begin{equation}
    \ket{\gt{m}} = 
        \sum_{\gt{m}_1 \in GT(\varphi_m^n)}
        \sum_{\gt{m}_2 \in GT(\bar\varphi_m^n)}
            C_{\gt{m}_1, \gt{m}_2}^{\gt{m}} \ket{\gt{m}_1, \gt{m}_2}.
\end{equation}
As $\gt{m}_1$ and $\gt{m}_2$ are GT-patterns associated with $\varphi_m^n$ and
its dual respectively, $\ket{\gt{m}_1, \gt{m}_2}$ is easily expressed as a Fock
basis element \cite{alex_numerical_2011,arienzo_bosonic_2025}. It will be
convenient to denote by $\hat{\gt{m}}$ the operator (i.e., in matrix form)
associated with the GT-pattern $\gt{m}$.

\section{Estimating the expectation value of observables}\label{sec:estExpVal}

In this section, we give an upper bound on the variance of the classical shadow
estimator, that allows us to derive sample complexity bounds for estimating the
expectation values of a collection $M$ observables. Then, we give the total
running time complexity of the algorithm. We show that both the sample
complexity and the running time of the post-processing scale with the dimension
of the largest irreducible space the overlap has support on. This observation,
where irreducible representation play a crucial role for the complexity
analysis, was already pointed out in
\cite{goh_liealgebraic_2025,sauvage_classical_2024}.

Let $\rho$ be an unknown quantum state over $m \geq 1$ modes and given a Haar
random unitary matrix $U$ and measurement outcome $\s \in \Phi_m^n$, recall
that we defined
\begin{equation}\label{eq:classicalShadow}
    \hat\rho_{(U, \s)} = \Minv{n}{\omega_{m}^{n\dagger}(U)(\Pi_\s)}
\end{equation}
as a \emph{classical snapshot} of $\rho$. It is useful to consider this object
for the purpose of the exposition, but it should be computed exactly in
practice. The reason is that it requires computing matrix permanents, which
takes exponential time \cite{ryser_combinatorial_1973}. Rather, we store the
pair $(U, \s)$ in classical memory. Recall that the measurement channel, as well
as its inverse, are self-adjoint \cite{sauvage_classical_2024}. Let $O$ be a
quantum observable and let $\Pc = \sum_{n\geq 0}\Pc^{(n)}$ be the
(super-)projector onto the fixed photon-number subspaces of the Hilbert space.

Indeed, by construction we have
\begin{equation}
    \Pc(\rho) = \e[U, \s]{\hat \rho_{(U, \s)}}.
\end{equation}

Then, by linearity of the trace and self-adjointness of the measurement channel,
we get
\begin{equation}\label{eq:expectationValueFromSnapshot}
    \begin{aligned}
    \tr{O\Pc(\rho)}
        & = \e[U, \s]{\tr{O\hat \rho_{U, \s}}} \\
        & = \sum_{n \geq 0}\e[U, \s]{\tr{O\Minv{n}{\omega_{m}^{n\dagger}(U)(\Pi_\s)}}} \\
        & = \sum_{n \geq 0}\e[U, \s]{\tr{\Minv{n}{O}\omega^{n\dagger}_m(U)(\ketbra{\s})}} \\
        & = \sum_{n \geq 0}\e[U, \s]{\bra{\s}\varphi^{n}_m(U)\Minv{n}{O}\varphi^{n\dagger}_m(U)\ket{\s}}. \\
    \end{aligned}
\end{equation}

For simplicity of the argument, we assume in the rest that $\rho$ has a fixed
photon-number $n$. The general case follows from
\cref{eq:expectationValueFromSnapshot} and summing from $n=0$ to the largest
number of photons measured throughout the sample collection phase.

\subsection{Estimator in the GT basis}
We denote by
\begin{equation}\label{eq:defFOUS}
    f_O(U, \s) = \bra{\s}\varphi^{n}_m(U)\Minv{n}{O}\varphi^{n}_m(U^\dag)\ket{\s}
\end{equation}
the estimator of interest in \cref{eq:expectationValueFromSnapshot}. The
inverse of the channel is readily found from its definition to be
\begin{equation}
    \begin{aligned}
        \Mc^{(n)^{-1}}
        & = \sum_{k = 0}^n \textstyle\frac{1}{s_{\lambda_k}} \Pi_{\lambda_k},
    \end{aligned}
\end{equation}
where $\Pi_{\lambda_k}$ projects onto the subspace carrier of $\lambda_k$ in the
Fock basis. That said, we can write

\begin{equation}\label{eq:pluggingChannelInF}
    \begin{aligned}
        f_O(U, \s) 
        & = \sum_{k = 0}^n \frac{1}{s_{\lambda_k}} \bra{\s} \varphi^{n}_m(U) \lpr 
                 \Pi_{\lambda_k} \lpr O \rpr 
             \rpr \varphi^{n}_m(U) \ket{\s} = \sum_{k = 0}^n \frac{1}{s_{\lambda_k}} \bra{\s} \omega^{n}_m(U)(\Pi_{\lambda_k} (O)) \ket{\s}. \\
    \end{aligned}   
\end{equation}
As $\Pi_{\lambda_0}$ is the projector onto the subspace carrier of the trivial
irrep, we have
\begin{equation}\label{eq:projTrivialIrrep}
    \Pi_{\lambda_0} \lpr O \rpr = \frac{\tr{O}}{\dim \Hc_{m, n}}\Id_{m, n},
\end{equation}
and all other $\Pi_{\lambda_k} \lpr O \rpr$ have range in the traceless
subspaces.

By linearity of the
expectation value,
\begin{equation}
    \e[U, \s]{f_O(U, \s)} = \sum_{k=0}^n\e[U, \s]{f_{\Pi_{\lambda_k}(O)}(U, \s)}.
\end{equation}
The functions $f_{\Pi_{\lambda_k}(O)}(U, \s)$ are convenient as they resemble
the filter functions of filtered benchmarking
\cite{arienzo_bosonic_2025,wilkens_benchmarking_2024}. This form is convenient
if $O$ has only support on a subset of the irreps of $\omega_m^n$. We exploit
this fact for determining both the variance of the estimator and the running
time of the protocol.

\subsection{Observables in the Fock basis}\label{app:observableFockBasis} In
this section, we discuss the relation between the degree of an observable and
the irreps onto which it has nonzero support. We define the generators $E_{i,j}
= a_i^\dagger a_j$ for $1 \leq i, j \leq m$. They satisfy $[\hat N, E_{i,j}] =
0$, where $\hat N = \sum_{k} \hat n_k$. It will be useful to consider their
higher-degree counterparts, which we write ${E}_{\bm i, \bm j} = a_1^{\dagger
i_1} \cdots a_m^{\dagger i_m}a_1^{j_1}\cdots a_j^{j_m}$ for $\bm i, \bm j \in
\Phi_m^k$ indexing modes and where $1\leq k\leq n$ identifies their degree.
Define $A^{(k)}$ as the span of degree-$k$ monomials $E_{\bm i, \bm j}$ from
which we define the filtered space
\begin{equation}
    \Ac_d = \bigcup_{k=0}^d A^{(k)},
\end{equation}
that is the space of polynomials in the $E_{\bm i,\bm j}$'s of degree at most
$d$. Then, we have the following chain of (strict) inclusions
\begin{equation}\label{eq:LHasFilteredUnion}
    \Ac_0 \subset \Ac_1 \subset \cdots \subset \Ac_n = \Ls(\Hc_m^n).
\end{equation}
We now prove the following \cref{pr:equivalenceDegreeIrrep} that unveils the
link between the degree of an observable and the irreps supporting it.
\begin{proposition}\label{pr:equivalenceDegreeIrrep} Let $O \in \Ls(\Hc_m^n)$ be
a quantum observable. The two statements are equivalent
\begin{enumerate}
    \item $O$ is a polynomial of irreducible degree $k$,
    \item $O$ has support on the $k$ first isotypic subspaces of $\Ls(\Hc_m^n)$
    (not counting the trivial subspace). 
\end{enumerate}
\end{proposition}

\begin{proof}
We prove both directions.  We will use the following fact about the structure of
the operator space, alike to \cite[Theorem 2]{wilkens_benchmarking_2024}. Let
$n_1, n_2>n_1 \in \N$. By the above results, we can write
\begin{equation}
    \Ls(\Hc_m^{n_1}) \simeq \bigoplus_{k=0}^{n_1} \Hc_{\lambda_k}
    \quad \text{and} \quad
    \Ls(\Hc_m^{n_2}) \simeq \bigoplus_{k=0}^{n_2} \Hc_{\lambda_k},
\end{equation}
but more importantly, as $n_2>n_1$ and $m$ is fixed, we have
\begin{equation}\label{eq:subspacesCoincide}
    \Ls(\Hc_m^{n_2}) \simeq \bigoplus_{k=0}^{n_1} \Hc_{\lambda_k} \oplus \bigoplus_{\ell=n_1}^{n_2} \Hc_{\lambda_\ell}.
\end{equation}
Informally, adding new photons adds new blocs to the isotypic decomposition of
the operator space and let the already present blocs unchanged. That is to say,
$\Ls(\Hc_m^{n_1})$ is \emph{contained} into $\Ls(\Hc_m^{n_2})$.

$(i) \Rightarrow (ii)$: Let $O$ be a polynomial of irreducible degree $n_1$,
i.e., there exists no $\tilde O$ of degree less than $n_1$ such that
$\expval*{O} = \expval*{\tilde O}$ always holds.
By definition, $O \in \Ac_k$, i.e., $O \in \Ls(\Hc_m^k)$, and using the
structure of the isotypic decomposition of $\Ls(\Hc_m^n)$ -- that is to say,
for $n_2>n_1$, the $n_1$ first isotypic subspaces of $\Ls(\Hc_m^{n_1})$ and
$\Ls(\Hc_m^{n_2})$ coincide (\cref{eq:subspacesCoincide}), we conclude
that $O$ has support on the $n_1$ first isotypic subspaces of $\Ls(\Hc_m^{n_2})$.

$(i) \Leftarrow (ii)$: Let $O \in \Ls(\Hc_m^{n_2})$ be a quantum observable such
that for all $j>n_1$, $\Pi_{\lambda_j}(O) = 0$. Similarly, using the structure
of the isotypic decomposition of $\Ls(\Hc_m^{n_2})$, it follows that $O \in
\Ls(\Hc_m^{n_1})$. As per \cref{eq:LHasFilteredUnion}, $\Ls(\Hc_m^{n_1})$ is
spanned by monomials of degree at most $n_1$ and the claim follows.
\end{proof}

\subsection{Variance of the estimator}\label{sec:varianceBound}

In this section, we state the second main result of this work. We claim that the
sample complexity of the protocol depends on both the $\infty$-norm of the
observables and its support with the irreps of $\omega_m^n$. First, we introduce
the \emph{photonic shadow-norm} which alike qubit-based shadows dictates the
sample complexity.

\begin{lemma}[Bound on estimators variance]\label{lm:variance} Fix
    $O\in\Ls(\Hc_m^n)$ Hermitian and let $\hat o = \tr{O\hat \rho^{(n)}}$ where
    $\hat\rho^{(n)}$ is a classical shadow (as defined in
    \cref{eq:classicalShadow}) of the $n$-photon component obtained upon
    projection of an unknown Fock state $\rho$. Then, 
    \begin{equation}
        \var{\hat o} = \e{(\hat o - \e{\hat o})^2} \leq \left\|O - \textstyle\frac{\tr{O}}{\dim \Hc_m^n}\Id_{m, n}\right\|_{\psn}^2,
    \end{equation}
    where the photonic shadow-norm only depends on the measurement primitive.
\end{lemma}
\begin{proof}
    Importantly, at most finite number of photons $N$ will be detected in
    experiments. As explained above, PNR measurements induce a projection onto a
    subspace of fixed photon-number given by the number of measured photons ---
    therefore, we assume that $O$ has no off-block-diagonal terms. First, the
    channel $\Mc$ (defined in \cref{eq:channel}) is self-adjoint, i.e.,
    ${\tr{\sigma^{\dagger}\Mc(\rho)} = \tr{\Mc(\sigma^{\dagger})\rho}}$. Let
    $\hat \rho^{(n)} = \Pc^{(n)}(\hat\rho)$ be the $n$-photon component of
    $\hat\rho$ for $n\leq N$, and define $\rho^{(n)}$ similarly. We will work
    with $n$ and the result generalises to the whole state by taking the direct
    sum of each component. Alike the qubit case, it holds that the variance only
    depends on the traceless part of the observable:
    \begin{equation}
        \begin{aligned}    
        \hat o - \e{\hat o} =  \tr{O\hat\rho^{(n)}} - \tr{O\rho^{(n)}} = \tr{O_0\hat\rho^{(n)}} - \tr{O_0\rho^{(n)}},
        \end{aligned}
    \end{equation}
    where $O_0$ is the traceless part of the observable $O$. The variance can be
    written as
    \begin{equation}
        \begin{aligned}
            \var{\hat o} 
                & = \e{(\hat o - \e{\hat o} )^2} \\
                & = \e{\tr{O_0\hat\rho^{(n)}}^2}- \tr{O_0\e{\hat\rho^{(n)}}}^2 \\
                & = \e{\tr{O_0\hat\rho^{(n)}}^2}- \tr{O_0 \rho^{(n)}}^2.
        \end{aligned}
    \end{equation}
    By self-adjointness of the measurement channel, expanding the expectation value
    over $\s$ yields
    \begin{equation}\label{eq:expSquaredTrace}
        \begin{aligned}
            \e{\tr{O_0\hat\rho^{(n)}}^2} 
                & = \e[U, \s]{\mel{\s}{\varphi_m^n(U)\Mc^{(n)^{-1}}(O_0)\varphi_m^n(U)^{\dagger}}{\s}^2} \\
                & = \e[U]{\sum_{\s \in \Phi_m^n} \bra{\s} \varphi_m^n(U)\rho\varphi_m^n(U)^\dagger \ketbra{\s}\varphi_m^n(U)\Mc^{(n)^{-1}}(O_0)\varphi_m^n(U)^{\dagger}\ket{\s}^2 }.
        \end{aligned}
    \end{equation}
    Alike the qubit case, we define the \emph{photonic shadow-norm} by
    maximizing \cref{eq:expSquaredTrace} over all possible states $\sigma$ in
    order to remove the dependence on the input state:
    \begin{equation}\label{eq:introductionPSNorm}
        \|O_0\|_{\psn} = \max_{\sigma: \text{state}} \lpr\e[U]{\sum_{\s \in \Phi_m^n} \bra{\s} \varphi_m^n(U)\sigma\varphi_m^n(U)^\dagger \ketbra{\s}\varphi_m^n(U)\Mc^{(n)^{-1}}(O_0)\varphi_m^n(U)^{\dagger}\ket{\s}^2}\rpr^{\frac{1}{2}}.
    \end{equation}
    Indeed, one can check that the photonic shadow norm $\|\cdot\|_{\psn}$ as
    defined in \cref{eq:introductionPSNorm} is indeed a valid norm.
\end{proof}

We provide an upper-bound to the photonic shadow-norm in the following \cref{lm:PSNUpperBound}
\begin{lemma}\label{lm:PSNUpperBound}
    Let $O \in \Ls(\Hc_m^n)$ Hermitian be an observable of degree $d$ in
    $E_{i,j}$. Then, it holds that
    \begin{equation}
        \|O\|_{\psn}^2 = \bigo{\|O\|_{\infty}^2m^{3d}},
    \end{equation}
    and the same bound holds when considering the traceless part of the
    observable.
\end{lemma}

\begin{proof}
    
For a fixed photon-number $n$, $\Mc^{(n)}$ admits a set of eigenvectors
$\{\hat{\gt{m}}\}_{\gt{m}\in GT(\lambda_k), 0\leq k \leq n}$, indexed by
GT-patterns, that forms an orthonormal basis of $\Ls(\Hc_{m, n})$. As
GT-patterns are a convenient way to label basis elements of the irreducible
representations of $\Hc_m^n \otimes \Hc_m^n \simeq \Ls(\Hc_m^n)$ in vectorised
form, we denote by $\hat{\gt{m}}$ the operator associated with the GT-pattern
$\gt{m}$. This implies that $\Mc^{(n)}$ as a quantum channel acts on
$\hat{\gt{m}}$ as $\Mc^{(n)}(\hat{\gt{m}}) = \beta_i\hat{\gt{m}}$. The set
$\Lambda_n = \{\beta_i\}_{1 \leq i \leq d_m^{n^2}}$ of eigenvalues of
$\Mc^{(n)}$ coincides with the set $\{s_{\lambda_k}\}_{0 \leq k \leq n}$ as
defined in \cref{app:eq:slambdak} taken with multiplicity: each
$s_{\lambda_k}$ has multiplicity $d_{\lambda_k}$ in $\Lambda_n$. The eigenvalues
of the measurement channel can thus be written
\begin{equation}\label{eq:ithEigenvalue}
    \beta_{\gt{m}}
        = \tr{\Mc^{(n)}(\hat{\gt{m}})\hat{\gt{m}}} 
        = \sum_{\s \in \Phi_m^n}\intg{\U(m)}{\tr{\omega_m^n(U)^{\dagger}(\Pi_{\s})\hat{\gt{m}}}}{\mu_H(U)}^2.
\end{equation}

Let $\hat e_{\gt{m}} = \tr{\hat\rho^{(n)}\hat{\gt{m}}}$ be an estimate obtained
from the $n$-photon snapshot $\hat\rho^{(n)}$. Using
\cref{eq:expSquaredTrace,eq:ithEigenvalue}, we bound the variance of the
estimate as
\begin{equation}
    \begin{aligned}    
        \|\hat{\gt{m}}\|_{\psn}^2 
            & = \max_{\sigma} \e[U] {\sum_{\s \in \Phi_m^n}
                \tr{\omega_m^n(U)^\dagger(\Pi_{\s})\sigma}
                \tr{\omega_m^n(U)(\Pi_{\s})\Mc^{(n)^{-1}}(\hat{\gt{m}})}^2} \\
            & \leq  \beta_{\gt{m}}^{-2}\intg{\U(m)}{{\sum_{\s \in \Phi_m^n}
                \tr{\omega_m^n(U)^{\dagger}(\Pi_{\s})\hat{\gt{m}}}^2}}{\mu_H(U)}\\
            & = \beta_{\gt{m}}^{-1}.
    \end{aligned}
\end{equation}
As the set $\{\hat{\gt{m}}\}_{\gt{m}}$ forms an orthonormal basis of
$\Ls(\Hc_m^n)$, any observable $O$ admits a decomposition in this basis of
the form of $O = \sum_{\gt{m}} \alpha_{\gt{m}}\hat{\gt{m}}$.
By combining the above results, we get the following generic bound:
\begin{equation}\label{eq:genericVarianceBound}
    \begin{aligned}
        \|O\|_{\psn}^2 
            \leq \sum_{\gt{m}}\alpha_{\gt{m}}^2 \|\hat{\gt{m}}\|_{\psn}^2
            \leq \sum_{\gt{m}}\alpha_{\gt{m}}^2 \beta_{\gt{m}}^{-1}.
    \end{aligned}
\end{equation}

Considering the traceless part of the observable $O_0 = O -
\frac{\tr{O}}{\dim\Hc_m^n}\Id$, by decomposing the observable along the irreducible representations, the above \cref{eq:genericVarianceBound} is
upper-bounded by:
\begin{equation}
    \|O_0\|_{\psn}^2
   \leq
    \sum_{k=0}^{\deg(O)} s_{\lambda_k}^{-1}\sum_{\gt{m} \in GT(\lambda_k)} \alpha_{\gt{m}}^2
    \leq
    \sum_{k=1}^{\deg(O)} s_{\lambda_k}^{-1} \|\Pi_{\lambda_k}(O_0)\|_2^2 
    \leq 
    \|O\|_{\infty}^2\sum_{k=1}^n s_{\lambda_k}^{-1}  \dim \Hc_{\lambda_k}
\end{equation}
where the second step
uses $\Pi_{\lambda_0}(O_0) = 0$. 
Recall that 
\begin{equation}
    \dim \Hc_{\lambda_k} = \frac{2k - m - 1}{m-1}\binom{k+m-2}{k}^2
    \qquad  \text{and} \qquad
    s_{\lambda_k}^{-1} = \frac{2k - m - 1}{m-1}\binom{k+m-2}{k}.
\end{equation}
Combining the two above identities, we get that for a degree $d$ observable $O$,
which we recall using \cref{pr:equivalenceDegreeIrrep} is such that
$\Pi_{\lambda_{d'}}(O) = 0$ for all $d'>d$,

\begin{equation}
\|O_0\|_{\psn}^2 \leq \|O_0\|_{\infty}^2 d \dim \Hc_{\lambda_d} s_{\lambda_d}^{-1}, 
\end{equation}
thus it holds that $\|O_0\|_{\psn}^2 = \bigo{\|O_0\|^2_{\infty}m^{3d}}$ and the
proof is complete.
\end{proof}

The main tool to characterize the sample complexity of the scheme is the
median-of-means estimators, which we recall in \cref{def:mom}.

\begin{definition}[Median-of-means estimator
\cite{huang_predicting_2020}]\label{def:mom} Let $X$ be a random variable with
variance $\sigma^2$. Then, $K$ independent sample means of size $N =
\bigo{\sigma^2/\varepsilon^2}$ suffice to construct a median-of-means estimator
$\mu(N, K)$ that obeys 
\begin{equation}
    \pr{|\mu(N, K) - \e{X}| \geq \varepsilon} \leq 2 e^{-K/2}
\end{equation}
for all $\varepsilon > 0$.
\end{definition}

We consider a classical shadow consisting in a collection of $NK$ snapshots
$\{\hat \rho_1, \cdots \hat \rho_{NK}\}$. The final estimate $\hat o_i$ is
obtained using the median-of-mean estimator:
\begin{equation}
    \hat o_i(N, K) = \mathrm{median}\aco \hat o_i^{(k)}(N, k)\acf_{k = 1}^K,
\end{equation}
where ${\hat o_i^{(k)}(N, k) = \frac{1}{N} \sum_{j =N(k-1) +1}^{Nk} \tr{o_i \hat
\rho_j}}$ for $1 \leq k \leq K$ is a sample mean obtained from $N$ independent
snapshots.

\begin{theorem}[Sample complexity]\label{thm:sampleComplexity} Given a
    collection $O_1, \cdots, O_T$ quantum observables and accuracy parameters
    $\varepsilon, \delta \in [0,1]$, set 
    \begin{equation}
        N = \bigo{\varepsilon^{-2} \max_{t} \| O_{t0}\|_{\psn}^2}
        \qquad \text{and} \qquad
        K = 2 \log(2T/\delta),
    \end{equation}
    then a classical shadow $\S(\rho, NK)$ allows, via the median-of-means
    estimator, for estimating with probability $1-\delta$ all properties
    $\tr{O_i\rho}$ to within additive precision $\pm\varepsilon$.
\end{theorem}

\begin{proof}
From \cref{lm:variance} one can compute an estimate of $\hat o_i(N, K)$ to
within additive error $\epsilon$ with probability $1-\delta$ by setting 
\begin{equation}
    N = \bigo{\varepsilon^{-2} \max_{t} \| O_{t0}\|_{\psn}^2
    },
\end{equation}
where the value of $N$ follows from the variance of the estimator and the performance bound of the median-of-mean
estimator. $K$ is chosen such that 
\begin{equation}
    \pr{ |\hat o_i(N, K) - \tr{O_i \rho} | \geq \varepsilon } \leq \frac{\delta}{T},
\end{equation}
i.e., obtain that merely $K = 2 \log(2T/\delta)$ independent estimates $\hat
o_i^{(k)}(N, k)$ are needed to compute an estimate of $o_i$ using the
median-of-means estimator (see \cref{def:mom}).
\end{proof}

\subsection{Time-complexity analysis}\label{app:timeComplexity}
The running time of the protocol is often overlooked, yet it constitutes a
critical bottleneck for the overall efficiency of classical shadow methods.
In this section, we provide two techniques for performing the post-processing.
The first technique detailed in \cref{app:approximationTechnique} is based on
the best known approximation algorithm \cite{lim_classical_2026}, for which the
main figure of merit is the eigenspectrum of the observables. The second method,
detailed in \cref{app:exactTechnique} is based on an algorithm for computing
exactly expectation values of observables in time scaling exponentially with
their degree.

This further stresses the role of low dimensional irreducible representations of
$\Ls(\Hc_m^n)$, as we show that the time complexity of the protocol scales
with the dimension of the irreps the observables have overlap with. Recall that
the goal is to collectively estimate the expectation value of $M$ observables
$\{O_1, \dots, O_M\}$. 

These techniques are independent and complementary and must be chosen depending
on the observables of interest. Nonetheless, the exact algorithm seems more
useful in practice as the approximate algorithm exhibits a time-dependency on
the observables $2$-norm which is likely exponential in the photon-number for
low-degree observables.

\subsubsection{Exact algorithm}\label{app:exactTechnique}

In this section, we prove with \cref{pr:degreeKObs} that the expectation
value of a degree-$k$ monomial can be computed exactly in time scaling
exponentially with $k$. This further implies that the expected value of
constant-degree observables can be computed in polynomial time, matching the
sample-efficient regime for the classical shadow protocol. This follows from the
fact that the number of independent degree-$k$ monomials is also exponential in
$k$.

\begin{proposition}\label{pr:degreeKObs} Let $\bm p, \bm q \in \Phi_m^k$ be
    multi-indices and define ${E}_{\bm p, \bm q}=a_1^{\dagger p_1}\cdots
    a_m^{\dagger p_m} a_1^{q_1} \cdots a_m^{q_m}$ to be a monomial observable of
    degree $k$, let $U \in \U(m)$ be a unitary matrix and $\bs n \in \Phi_m^n$,
    i.e., $\ket{\bs n}$ is a Fock basis state. Denote by $\ket{\psi} =
    \varphi_m^n(U)\ket{\bs n}$. Then, $\expval*{{E}_{\bm p, \bm
    q}}_{\ket{\psi}}$ can be computed exactly in time $\bigo{k2^k\binom{\min(k,
    n-k) + s - 1}{s - 1}}$ with $s = |\{i \ |\ n_i > 0\}|$ via
    \begin{equation}
    \expval*{{E}_{\bm p, \bm q}}_{\ket{\psi}}  
        = \sum_{\substack{\bm \ell \in \Phi_m^k \\ \bm \ell \preccurlyeq \bm n}}\frac{\bm n!}{(\bm n - \bm \ell)!\bm\ell!^2} \per{U_{\bm p, \bm \ell}}^* \per{U_{\bm q, \bm \ell}},
    \end{equation}
    where $\bm \ell \preccurlyeq \bm n$ is the partial entry-wise order.
\end{proposition}

\begin{proof}
We introduce vector $\vec i, \vec j \in \{1, \cdots, m\}^k$, such that $E_{\bm
p, \bm q} = a_{i_1}^\dagger \cdots a_{i_k}^\dagger a_{j_1}\cdots a_{j_k}$. For
such vectors, we introduce the function $\varDelta_{\vec i, \vec j}$ defined as 
\begin{equation}
    \varDelta_{\vec i, \vec j} = \begin{cases}
        1 & \text{if } \exists\ \sigma \in S_k: \vec i = \sigma(\vec j),\\
        0 & \text{otherwise},
    \end{cases}
\end{equation}
that is to say, $\varDelta_{\vec i, \vec j} = \delta_{\bm p, \bm q}$. Now,
observe that as ${E}_{\vec i, \vec j} = \prod_{\alpha=1}^m
a_{i_\alpha}^{\dagger} a_{j_\alpha}$, it holds that
\begin{equation}\label{eq:expValHopping}
        \expval*{{E}_{\bm p, \bm q}}_{\ket{\bs n}}
            = \delta_{\bm p, \bm q} \frac{\n!}{(\n - \bm p)!} 
            = \varDelta_{\vec{i}, \vec{j}} \prod_{\substack{\alpha = 1\\ 0 \leq \#_{\vec i}(\alpha) \leq n_{i_\alpha}}}^m\frac{n_{i_\alpha}!}{(n_{i_\alpha} - \#_{\vec i}(\alpha))!},
\end{equation}
where we defined $\#_{\vec i}(\alpha)$ as the number of $\alpha$'s in $\vec i$,
i.e., $\#_{\vec i}(\alpha) = \sum_{\beta = 1}^k \delta_{i_\beta, \alpha}$.
\begin{subequations}
    \begin{align}
        \expval*{{E}_{\bm p, \bm q}}_{\ket{\psi}} 
            & = \expval*{\prod_{\alpha = 1}^k \lpr\sum_{x_\alpha = 1}^m U_{i_\alpha, x_\alpha}^*a^\dagger_{x_\alpha} \rpr
                \prod_{\beta = 1}^k \lpr \sum_{y_\alpha = 1}^m  U_{j_\beta, y_\beta}a_{y_\beta}\rpr }_{\ket{\bs n}} \label{subeq:line1}\\
            & = \sum_{\vec x, \vec y \in \{1, \cdots, m\}^k} \expval*{\prod_{\alpha, \beta = 1}^k  U_{i_\alpha, x_\alpha}^* U_{j_\beta, y_\beta}a^\dagger_{x_\alpha} a_{y_\beta} }_{\ket{\bs n}}\\
            & = \sum_{\vec x, \vec y \in \{1, \cdots, m\}^k} \prod_{\alpha, \beta = 1}^k  U_{i_\alpha, x_\alpha}^* U_{j_\beta, y_\beta}\expval*{ a^\dagger_{\vec x} a_{\vec y} }_{\ket{\bs n}}\\
            & = \sum_{\vec x, \vec y \in \{1, \cdots, m\}^k} \varDelta_{\vec x, \vec y}  \prod_{\alpha, \beta = 1}^k  U_{i_\alpha, x_\alpha}^* U_{j_\beta, y_\beta} \prod_{\gamma=1}^m \frac{n_\gamma!}{(n_\gamma - \#_{\vec \ell}(\gamma))!}\label{subeq:line4}\\
            & = \sum_{\vec x \in \{1, \cdots, m\}^k}  \prod_{\alpha = 1}^k  U_{i_\alpha, x_\alpha}^* \lpr \sum_{\sigma \in S(\vec x)}  \prod_{\beta = 1}^k U_{j_\beta, \sigma(\vec x)_\beta} \rpr \prod_{\gamma=1}^m \frac{n_\gamma!}{(n_\gamma - \#_{\vec x}(\gamma))!}\label{subeq:line5}\\
            & = \sum_{\ell_1\leq \cdots \leq \ell_k = 1}^m \sum_{\tau \in S(\vec \ell)}\prod_{\alpha = 1}^k  U_{i_\alpha, \tau(\vec \ell)_\alpha}^* \lpr \sum_{\sigma \in S(\tau(\vec \ell))}  \prod_{\beta = 1}^k U_{j_\beta, \sigma \circ \tau(\vec \ell)_\beta} \rpr \prod_{\gamma=1}^m \frac{n_\gamma!}{(n_\gamma - \#_{\vec \ell}(\gamma))!}\\
            & = \sum_{\ell_1\leq \cdots \leq \ell_k = 1}^m \lpr \sum_{\tau \in S(\vec \ell)} \prod_{\alpha = 1}^k  U_{i_\alpha, \tau(\vec \ell)_\alpha}^* \rpr \lpr \sum_{\sigma \in S(\vec \ell)}  \prod_{\beta = 1}^k U_{j_\beta, \sigma(\vec \ell)_\beta} \rpr \prod_{\gamma=1}^m \frac{n_\gamma!}{(n_\gamma - \#_{\vec\ell}(\gamma))!}\label{subeq:line7}.
    \end{align}
\end{subequations}

\cref{subeq:line1} expands the evolution of creation and annihilation
operators under unitary transformation (see
\cref{eq:evolutionCreationOperators}). \cref{subeq:line4} uses the
identity of \cref{eq:expValHopping}. In \cref{subeq:line5}, the delta
function is simplified and $S(\vec x)$ ranges over all permutations of $\vec x$.
Finally, \cref{subeq:line7} follows from the fact that ${S(\vec \ell) =
S(\sigma(\vec\ell))}$ for any permutation $\sigma \in S(\vec \ell)$. 

Recall that $\vec i, \vec j$ are obtained from two multi-indices $\bm p, \bm q$.
A careful review of the expressions in the brackets of \cref{subeq:line7}
reveals the formula of the matrix permanent and the above reduces to
\begin{equation}\label{eq:expValPerm}
    \expval*{{E}_{\bm p, \bm q}}_{\ket{\psi}}  
    = \sum_{\substack{\bm \ell \in \Phi_m^k \\ \bm \ell \preccurlyeq \bm n}}\frac{\bm n!}{(\bm n - \bm \ell)!\bm\ell!^2} \per{U_{\bm p, \bm \ell}}^* \per{U_{\bm q, \bm \ell}}.
\end{equation}

The matrices $U_{\bm p, \bm \ell}$ and $U_{\bm q, \bm \ell}$ are $k\times k$
matrices whose permanent can be computed exactly in time $\bigo{k2^k}$ using
Ryser's method \cite{ryser_combinatorial_1973}. For large $k$, it is still
possible to compute the permanent exactly if the matrix has small rank using
Barvinok's algorithm \cite{barvinok_two_1996}. 

Indeed, there are at most $|\Phi_m^k| = \binom{m+k-1}{k}$ such matrix
permanents. However, this bound can be very loose, as for instance if $\n = (n,
0, \dots, 0)$, only $\bm \ell = (k, 0, \dots, 0)$ satisfies $\bm \ell
\preccurlyeq \n$; this nonetheless suffices for our approach as we are limited
to $k = O(1)$.  As a final remark, we observe that $\bm p = \bm q$ allows us to
recover the result of \cite{mayer_counting_2011}.
\end{proof}

Now, we show how the projection of an operator in a given irrep is found as a
linear combination of fixed degree monomials. Formally, we prove the next
\cref{lm:shadowEstimateCstDegree} using the following technical
\cref{pr:individualProj}.
\begin{proposition}\label{pr:individualProj} For $k>0$, it holds that
    \begin{equation}\label{eq:recursiveProj}
        \Pi_{\lambda_k}(O) = 
            \sum_{{E}_{\bm p, \bm q} \in \Bc_{m, k}} 
                \tr{\lpr O - \sum_{\ell = 0}^{k-1}\Pi_{\lambda_{\ell}}(O)\rpr^\dagger {E}_{\bm p, \bm q}}\mathfrak{E}_{\bm p, \bm q},
    \end{equation}
    where $\Bc_{m, k}$ is the set of all independent degree $k$ monomials,
    satisfying $\abs{\ \Bc_{m, k}} = \dim \Hc_{\lambda_k} = \binom{m+k-1}{k}^2 -
    \binom{m+k-2}{k-1}^2$, and $\{\mathfrak{E}_{\bm p, \bm q}\}_{\bm p, \bm q}$
    is a set of operators satisfying the orthonormality relation
    $\tr{\mathfrak{E}_{\bm i, \bm j}{E}_{\bm k, \bm \ell}}=\delta_{\bm i, \bm
    k}\delta_{\bm j, \bm \ell}$.
\end{proposition}
\begin{proof}
First, from \cref{eq:projTrivialIrrep} we have the projection onto the
trivial irrep satisfies
\begin{equation}
    \Pi_{\lambda_0} \lpr O \rpr = \frac{\tr{O}}{\dim \Hc_{m, n}}\Id_{m, n}.
\end{equation}
We observe that $O - \sum_{\ell = 0}^{k-1}\Pi_{\lambda_{\ell}}(O)$ has indeed no
support on the first $k$ isotypic subspaces $\Hc_{\lambda_0}, \cdots,
\Hc_{\lambda_{k-1}}$ of $\Ls(\Hc_m^n)$. Then, from
\cref{pr:equivalenceDegreeIrrep} follows that any $E_{\bm p, \bm q} \in
\Bc_{m, k}$ is such that $\Pi_{\lambda_{k'}}(E_{\bm p, \bm q}) = 0$ for all $k'
> k$. As we note that $\Bc_{m, k}$ spans $\Hc_{\lambda_k}$ as it is a basis of
$\Ls(\Hc_m^k)$, the sum of \cref{eq:recursiveProj} runs over polynomially
many summands. However, $\Bc_{m, k}$ does not form an orthonormal basis. In this
light, let $\{\mathfrak{E}_{\bm p, \bm q}\}_{\bm p, \bm q}$ be a set of
operators satisfying $\tr{\mathfrak{E}_{\bm i, \bm j}{E}_{\bm k, \bm
\ell}}=\delta_{\bm i, \bm k}\delta_{\bm j, \bm \ell}$, such that any operator
can be decomposed as
\begin{equation}
    \rho = \sum_{\bm p, \bm q} \tr{{E}_{\bm p, \bm q} \rho}\mathfrak{E}_{\bm p, \bm q}.
\end{equation}

Now, a word on the time complexity on computing the traces. Let $\Gamma_{\bm p}
= \{ t : p_t > 0 \}$ be the set of distinct modes a diagonal operator ${E}_{\bm
p, \bm p}$ acts on, satisfying ${\lb\Gamma_{\bm p} \rb = d \leq \deg {E}_{\bm
p, \bm p}}$. Let us reorder the elements of $\bm p$ in such a way that the
nonzero elements appear first and denote by $\tilde{\bm p}$ the nonzero
elements, e.g., $\bm p = (0, 2, 0, 1)$ becomes $(2, 1, 0, 0)$ and $\tilde p =
(2, 1)$. This step amount to permuting the modes. Then, thanks to
\cref{eq:expValHopping}, we find that
\begin{equation}
    \begin{aligned}
        \tr{{E}_{\bm p, \bm q}}
        & = \sum_{\bs n \in \Phi_m^n} \mel{\bs n}{{E}_{\bm p, \bm q}}{\bs n} \\
        & = \delta_{\bm p, \bm q} \sum_{\bs n \in \Phi_m^n} \frac{\bm n!}{(\bm n - \bm p)!}\\
        & = \delta_{\bm p, \bm q} \sum_{s = 1}^n  \binom{m-d + n-s - 1}{n-s} \sum_{\tilde{\bs n} \in \Phi_d^s} \frac{\tilde{\bs n}!}{(\tilde{\bs n} - \tilde{\bs p})!}
    \end{aligned}
\end{equation}
where the last step accounts for all possible ways of occupying the remaining
$m-d$ modes by $n-s$ photons. There is a total of $\sum_{s=1}^n \lb \Phi_d^s \rb
= \binom{n+d}{d} - 1$ summands, which makes the computation time of
$\tr{{E}_{\bm p, \bm q}}$ polynomial in $n$ (of degree $d$) when the degree of
the monomial is constant (as this implies $d = O(1)$).

\end{proof}

\begin{lemma}\label{lm:shadowEstimateCstDegree} Let $O$ be a constant-degree
    observable. Then, the shadow estimator $\hat o = \tr{O \hat\rho_{(U_i, \bs
    s_i)}}$ can be computed exactly in polynomial-time.
\end{lemma}
\begin{proof}
    The self-adjointness property of the channel allows us to apply its inverse
    to the observable rather than $\varphi_m(U_i)^\dagger\ketbra{\bs
    s_i}\varphi_m(U_i)$. We combine \cref{pr:degreeKObs,pr:individualProj} to
    give a recipe for computing $\tr{O \hat\rho_{(U_i, \bs s_i)}}$ scaling
    exponentially with the degree of the observable and polynomially with the
    photon-number and modes. Expanding the action of the channel along each
    isotypic subspace of $\Ls(\Hc_m^n)$ yields
    \begin{equation}
        \begin{aligned}
            \tr{O \hat\rho_{(U_i, \bs s_i)}}
            = \expval*{\Mc^{(n)^{-1}}(O)}_{\varphi_m^{n}(U_i^\dagger)\ket{\bs s_i}}
            = \sum_{k = 0}^{\deg O} s_{\lambda_k}^{-1}\expval{\Pi_{\lambda_k}(O)}_{\varphi_m^{n}(U_i^\dagger)\ket{\bs s_i}}.
        \end{aligned}
    \end{equation}
    \cref{pr:individualProj} states that $\Pi_{\lambda_k}(O)$ is supported by
    at most $\dim \Hc_{\lambda_k}$ degree-$k$ monomials, the expectation value
    of each of which can be computed efficiently if $k = O(1)$ from
    \cref{pr:degreeKObs}. Moreover, observe that $\dim \Hc_{\lambda_k}$ is
    polynomial (in $m$) when $k = O(1)$, i.e., there is a polynomial number of
    such expectation values to compute. 
\end{proof}

\subsubsection{Approximation algorithm}\label{app:approximationTechnique}

A useful result to classically estimate expectation value of operators in a
linear optical is stated in \cref{pr:estimateLOObervables}.

\begin{proposition}[Approximating observables expectation values in linear
optics \cite{lim_classical_2026}]\label{pr:estimateLOObervables} Consider an
$m$-mode linear optical circuit $U$ and a product observable $O$. Then, for an
input product state $\ket\psi$ the expectation value
$\tr{OU\ketbra*{\psi}U^{\dagger}}$ can be approximated within additive error
$\varepsilon$ w.p. $1-\delta$ in time
\begin{equation}
    \bigo{\frac{m^2}{\varepsilon^2} \log\frac{1}{\delta}\|O\|_2^2}.
\end{equation}
\end{proposition}
The time complexity stated in \cref{pr:estimateLOObervables} can be improved
for observables acting nontrivially on only a subsystem of $\varphi_m^n(U)\ket{\psi}$.
Nonetheless, it allows us to derive the following \cref{lm:approxhatoi}.

\begin{lemma}[Approximation of $\hat o_i(N, K)$]\label{lm:approxhatoi}
    One gets an additive error approximation of $\hat o_i(N, K)$ to within
    $\pm \varepsilon$ with probability $1 - \delta$ in time 
    \begin{equation}
        \bigo{NK\lpr \frac{mn}{\varepsilon s_{\lambda_{\deg(O_i)}}} \rpr^2
            \log\lpr\frac{n}{\delta}\rpr\max_{t, \ell} \|\Pi_{\ell}(O_t)\|_2^2}.
    \end{equation}
\end{lemma}
\begin{proof}
    
Considering \cref{eq:pluggingChannelInF},  we can write 
\begin{equation}
   \begin{aligned}
    \hat o_i^{(k)}(N, k) 
        & = \frac{1}{N} \sum_{j =N(k-1) +1}^{Nk} \tr{O_i \hat \rho_j} \\
        & = \frac{1}{N} \sum_{j =N(k-1) +1}^{Nk} \sum_{\ell = 0}^n \frac{1}{s_{\lambda_\ell}}\tr{\omega_m^n(U_j)(\Pi_{\lambda_\ell} (O_i)) \ketbra{\s_j}}, \\ 
   \end{aligned} 
\end{equation}
where $\hat\rho_j$ is assumed to be the snapshot associated with unitary matrix
$U$ and measurement outcome $\ket{\s}$. As $\Pi_{\lambda_0}(O)$ admits a closed
form of any $O$ (see \cref{eq:projTrivialIrrep}) estimating $f_O(U, \s)$
using \cref{pr:estimateLOObervables} reduces to approximating $N$ collections
of $n$ expectation value of observables (each being a projection of the
observable in the isotypic subspaces of $\Ls(\Hc_m^n)$). 

We use \cref{pr:estimateLOObervables} to estimate individually each $\tr{O_i
\hat \rho_j}$ for all $i$ and $j$. We explained in with
\cref{pr:individualProj} how each projection can be obtained in a recursive
way. This allows us to get an estimate of the inner sum to within
$\pm\varepsilon n/s_{\lambda_n}$. Therefore, the error in each approximation
must grow like $\varepsilon s_{\lambda_{\deg(O_i)}}/n$. Hence, getting a $\pm
\varepsilon$ approximation of $\hat o_i^{(k)}(N, k)$ takes classical time
\begin{equation}
    \bigo{Nn \lpr \frac{mn}{\varepsilon s_{\lambda_{d_{\max}}}} \rpr^2
        \log\lpr\frac{1}{\delta}\rpr\max_{i, \ell} \|\Pi_{\ell}(O_i)\|_2^2},
\end{equation}
where $d_{\max} = \max_t\deg(O_t)$. To conclude, observe that this must be
applied to all the $K$ means required in the median-of-means estimator.
\end{proof}

The dependency in $1/s_{\lambda_{d_{\max}}}$ implies an exponential running time
in the general case as $1 \leq d_{\max} \leq n$. With the same reasoning, this
running becomes polynomial (in $n$ and $m$) whenever the observables only have
components in the low-dimensional isotypic subspaces of $\Ls(\Hc_m^n)$.

\section{Pseudo-photon number resolving measurement}\label{sec:ppnr}

\paragraph{Principle of Pseudo-PNR measurement.} The protocol we present assumes
access to photon-number detectors, i.e., detectors that output the occupation
number of the measured mode. However, the complexity of building such detectors
grows with the maximal photon-number they are expected to resolve. Simpler
detectors---threshold detectors---distinguishing between the vacuum and one or
more photons are more widespread. It is possible to mimic the behaviour of PNR
detectors using threshold detectors in a larger interferometer as follows. A
small $p$-mode interferometer $\Fc_p$ is appended to each of the $m$ output
modes, thus forming a new $M=mp$ mode interferometer 
\begin{equation}
    \tilde U = (\Fc_p)^{\otimes m} P (U \otimes \Id_{(m-1)p}),
\end{equation}
where $P$ is a suitable permutation that ensures that each output modes of $U$
becomes the top mode of one of the $\Fc_p$. The $(m-1)p$ additional modes are
filled with the vacuum, so that the new $M$-mode input state is $\rho \otimes
\ket{0}^{(m-1)p}$. An example is shown in \cref{fig:ascella} for $p=3$ and
$m=4$.

As threshold detectors only distinguish whether the mode is occupied, the
purpose of $\Fc_p$ is to maximize the spread of the input photons among the $p$
modes. Formally, for an input state $\ket{n} \otimes \ket{0}^{\otimes p-1}$, the
goal is for $\Fc_p$ to maximize 
\begin{equation}
    g_{p, n} = \sum_{\substack{\bs b \in \{0,1\}^p \\ |\bs b| = n}} 
    \abs{\mel*{\bs b}{\varphi_p^n(\Fc_p)}{n, 0, \cdots, 0}}^2,
\end{equation}
where the sum ranges over all detection events where the threshold detectors
clicked for a single photon. This indeed enforces the constraint $p \geq n$. In
the case $p=n$, $g_{p, n}$ is maximised by the interferometer implementing the
discrete Fourier transform defined as 
\begin{equation}
    \Fc_p = \frac{1}{\sqrt{p}}(\omega^{kl})_{0\leq k,l\leq p-1},
\end{equation}
with $\omega = e^{-2\imath \pi/p}$ is a primitive $p$-th root of unity
\cite{aaronson_generalizing_2014}. We use this interferometer even for the cases
$p \neq n$,
in which case according to \cref{app:eq:bsamplitudes} 
\begin{equation}\label{app:eq:scalingFactorPPNRT}
    g_{p, n} = \binom{p}{n} \frac{n!}{p^n}.
\end{equation}

Each PNR detector is therefore emulated by several threshold detectors and a
Fourier interferometer. Upon measuring a bit-string $\bs b$ at the output of the
Fourier interferometer, the occupation for the mode corresponding to this
pseudo-PNR circuit is $|\bs b|$. However, the fact that $g_{p,n} \neq 1$ implies
that some measurement outcomes are discarded as threshold detectors cannot
resolve all occupations. This results in a distribution---which refer to as the
pseudo-PNR distribution---that is biased with respect to the actual
distribution. We now described how this bias can be mitigated.

With PPNR, however, we can only reconstruct a single block $\Pc^{(n)}(\rho)$
which induces an additional constraint on the photon-number. This constraint is
a result of the experimental apparatus we use for the demonstration rather than
a limitation of the protocol itself. More precisely, if $\rho$ is a
superposition of different photon-number states, the only block one can
construct using pseudo-PNR is that corresponding to the maximal photon-number.
This comes from the fact that with pseudo-PNR, it is not possible to distinguish
between the measurement of an actual $(n-1)$-photon Fock state and two photons
clicking the same threshold detector (yielding $n-1$ clicks on an $n$-photon
state) when $n$-photon can be expected.

\paragraph{Mitigating pseudo-PNR distributions.}
We fix $n$ without loss of generality and assume that $p \geq n$ as otherwise
some output states, those with occupation larger than $p$, cannot be resolved.
To describe the bias, we show that for any two $\s, \t \in \Phi_m^n$, 
\begin{equation}
\abs{\mel{\s}{\varphi_m^n(U)}{\t}}^2\prod_{1 \leq i \leq m} g_{p, s_i} = 
\sum_{\substack{\bs b_1, \cdots, \bs b_m \in \{0, 1\}^{p}\\ |\bs b_i| = s_i}} \abs{\mel*{(\bs b_1, \cdots, \bs b_m)}{\varphi_M^n(\tilde U)}{\bs t, 0^{\otimes (m-1)p}}}^2,
\end{equation}
i.e., that each probability of the pseudo-PNR distribution is biased by a known
factor. In order to remove the explicit application of $\varphi_m$ we note that
as it is a group homomorphism, it satisfies $\varphi_m(AB) =
\varphi_m(A)\varphi_m(B)$ and $\varphi_{mp}(A^{\otimes m}) =
\varphi_p(A)^{\otimes m}$.
\begin{equation}\label{eq:derivationScalingFactorPPNR}
    \begin{aligned}
        \sum_{\substack{\bs b_1, \cdots, \bs b_m \in \{0, 1\}^{p}\\ |\bs b_i| = s_i}} 
            \abs{\mel*{\bs b_1, \cdots, \bs b_m}{\tilde U}{\bs t, 0^{\otimes (m-1)p}}}^2 
        & = \sum_{\substack{\bs b_1, \cdots, \bs b_m \in \{0, 1\}^{p}\\ |\bs b_i| = s_i}} 
            \abs{\mel*{\bs b_1, \cdots, \bs b_m}{\Fc_p^{\otimes m} P (U \otimes \Id_{(m-1)p})}{\bs t, 0^{\otimes (m-1)p}}}^2 \\
        & =  \sum_{\substack{\bs b_1, \cdots, \bs b_m \in \{0, 1\}^{p}\\ |\bs b_i| = s_i}} 
            \prod_{1 \leq i \leq m}\frac{s_i!}{p^{s_i}}\abs{\mel*{\bs s, 0^{\otimes (m-1)p}}{U \otimes \Id_{(m-1)p}}{\bs t, 0^{\otimes (m-1)p}}}^2 \\
        & =  \abs{\mel*{\bs s}{U}{\bs t}}^2
            \sum_{\substack{\bs b_1, \cdots, \bs b_m \in \{0, 1\}^{p}\\ |\bs b_i| = s_i}} 
                \prod_{1 \leq i \leq m}\frac{s_i!}{p^{s_i}} \\
        & =  \abs{\mel*{\bs s}{U}{\bs t}}^2 \prod_{1 \leq i \leq m}g_{p, s_i}.\\
    \end{aligned}
\end{equation}
In the second step, we use the fact that $\abs{\mel{\bs
b_i}{\Fc_p}{s_i,0^{\otimes p-1}}}^2 = \frac{s_i!}{p^{s_i}}$ since $|\bs b_i| =
s_i$ by definition, together with the fact that $P$ is designed to match all but
the first modes of $\bra{\bs b_i}\Fc_p$ with the vacuum. The last step, we
use the fact that $\frac{s_i!}{p^{s_i}}$ is independent of $\bs b_i$, i.e., it
does not depend on the arrangements of zeros and ones in it, and the fact that
there are $\binom{p}{s_i}$ bitstrings of size $p$ of Hamming-weight $s_i$.
Therefore, each probability of the pseudo-PNR distribution is biased by a factor
\begin{equation}
    g_{p, \bs s} = \prod_{1 \leq i \leq m} g_{p, s_i}.
\end{equation}
The true distribution is recovered by multiplying each individual probability by
this factor.

\paragraph{Using mitigated pseudo-PNR distributions for classical shadows.}
We sample $T$ samples from the true distribution using pseudo-PNR measurements
as follows. First, sample $T_0$ times the output of the quantum circuit using
pseudo-PNR measurements to get a collection $\mathscr C_0 = (\bs s_1, \dots, \bs
s_{T_0})$. Construct a new collection $\mathscr C_1$ from $\mathscr C_0$ by
putting $g_{p, \bs s_i}^{-1}$ copies of the $i$-th sample $\bs s_i$: 
\begin{equation}
    \mathscr C_1 = (\underbrace{\bs s_1, \cdots, \bs s_1}_{g_{p, \bs s_1}^{-1}},\underbrace{\bs s_2, \cdots, \bs s_2}_{g_{p, \bs s_2}^{-1}}, \cdots, \underbrace{\bs s_{T_0}, \cdots, \bs s_{T_0}}_{g_{p, \bs s_{T_0}}^{-1}}).
\end{equation}
Now, sample $T$ elements of $\mathscr C_1$ uniformly at random with replacement.
Asymptotically (as $T_0$ grows), the empirical sample distribution of $\mathscr
C_1$ converges (in TVD) to the true distribution. We illustrate with an example
in \cref{fig:TVDPPNRExperiment}. This indeed implies a sample overhead as
$T_0$ needs to be large. Said otherwise, even single-shot shadows will require
many shots in order to sample form the right distribution using pseudo-PNR.

\paragraph{Mixing PPNR with restricted PNR}
We extend the method presented above of to measurements performed with PNR
resolving few photons. The mitigation technique boils down to deriving the
suitable scaling factor alike \cref{app:eq:scalingFactorPPNRT}. It is
obtained by considering not only bitstrings in the sum of
\cref{eq:derivationScalingFactorPPNR}, but all tuples whose entries are
smaller than the PNR resolution. 

We write $\lambda \vdash n$ a partition of $n$, i.e., a tuple $\lambda =
(\lambda_1, \cdots, \lambda_\ell)$ satisfying $\sum_{i=1}^\ell \lambda_i = n$
and $\lambda_1 \geq \lambda_2 \geq \cdots \geq \lambda_\ell \geq 0$. We denote
by $\ell(\lambda)$ its length and $\lambda \vdash_{k}n$ a partition of $n$ of
length $k$. Counting the possible zeros in $\lambda$, each $\lambda$ is uniquely
associated with a vector $\alpha(\lambda)$ satisfying $n = \sum_{i=0}^n i
\alpha(\lambda)_i$ -- $\alpha(\lambda)$ simply accounts for the multiplicities
in $\lambda$. For instance, the partition $\lambda = (2,2,1,0) \vdash_4 5$ is
associated with $\alpha(\lambda) = (1,1,2)$, as $5 = 1 \times 0 + 1 \times 1 + 2
\times 2$.

Partitions of $n$ of length $m$ allow us to identify elements of $\Phi_m^n$
up to permutations. Therefore, a measurement outcome can be seen a permuted
partition of the total number of measured photons of length equal to the number
of modes. Interestingly, the probability of resolving the output of the Fourier
interferometer only depends on the partition and not the permutation (assuming
all detectors have the same resolution). Given that, the zeros in the partitions
are important.

The probability for the output of a Fourier interferometer to be properly
resolved when measured with PNR detectors with resolution $r$ is, for the input
state $\ket*{n, 0^{\otimes (p-1)}}$, 
\begin{equation}
    h_{p, n, r} = \sum_{\substack{\lambda\ \vdash_p n \\ \max_i \lambda_i \leq r}}
    \binom{p}{\alpha(\lambda)} \abs{\mel*{\lambda}{\Fc_p}{n, 0^{\otimes (p-1)}}}^2 = 
    \sum_{\substack{\lambda\ \vdash_p n \\ \max_i \lambda_i \leq r}}
    \binom{p}{\alpha(\lambda)} \frac{n!}{p^n \lambda_1!\cdots\lambda_p!},
\end{equation} 
where $\binom{p}{\alpha(\lambda)}$ is the multinomial coefficient defined as
\begin{equation}
    \binom{n}{k_1, \cdots, k_m} = \frac{n!}{k_1! \cdots k_m!}
\end{equation} 
that counts all distinct permutations of $(k_1, \cdots, k_m)$. An illustrative
example is shown in \cref{fig:ppnrEfficiency}. Following the same lines as
\cref{eq:derivationScalingFactorPPNR} each probability of the pseudo-PNR
distribution performed with PNR detectors with resolution $r$ is biased by a
factor
\begin{equation}\label{eq:biasPPNRResolutionR}
    h_{p, \bs s, r} = \prod_{1 \leq i \leq m} h_{p, s_i, r}.
\end{equation}
\cref{eq:biasPPNRResolutionR} can be tweaked to account for detectors with
different resolutions by considering a vector ${\bs r = (r_1, \cdots, r_M)}$
filled with the resolution of each detector.

\begin{figure}[!ht]
    \includegraphics[width=.55\textwidth]{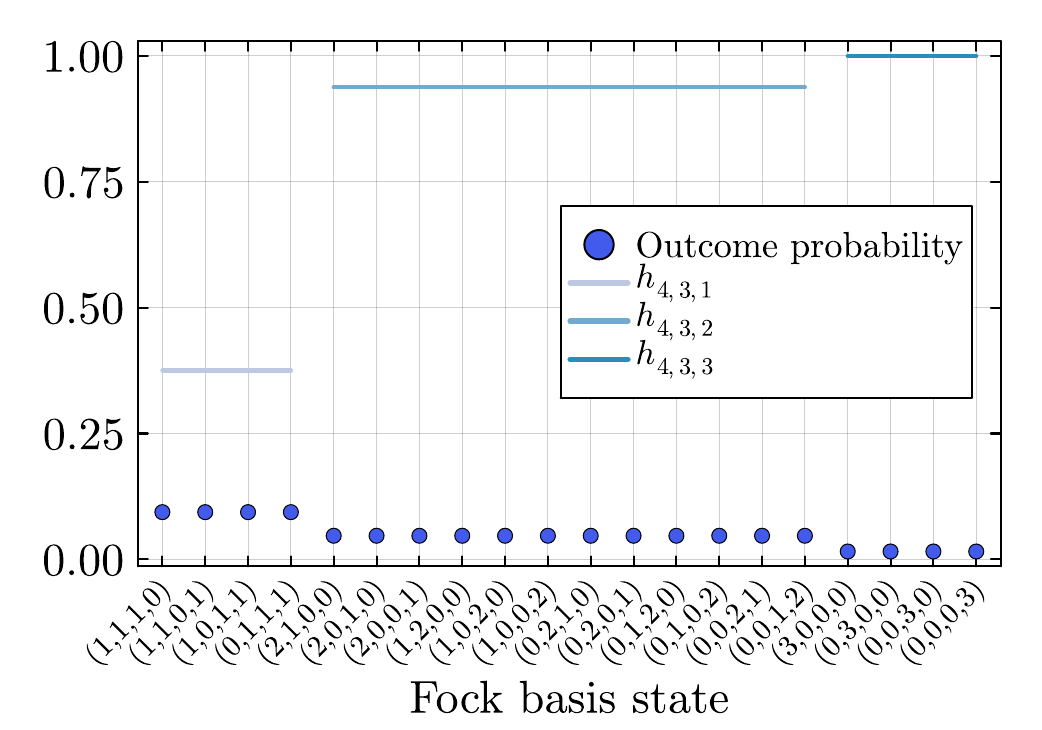}
    \caption{\justifying \textbf{Resolution efficiency of pseudo-PNR.}
    Illustration of different values of $h_{p,n,r}$ for a $4$ mode Fourier
    interferometer with input state $\ket{3, 0, 0, 0}$. The higher the
    resolution $r$ of the detectors, the better the efficiency of pseudo-PNR as
    it directly depends on the proportion of states that can be resolved
    unambiguously.}
    \label{fig:ppnrEfficiency}
\end{figure}

\section{Numerical simulations: the Bose-Hubbard Hamiltonian}\label{sec:BoseHubbard}
We demonstrate that the classical shadow protocol we introduce in this work can
efficiently predict properties of photonic states not produced by Fock basis
states evolved through a linear optical network. Moreover, this section serves
as an example of how the decomposition of an operator along the isotypic
subspaces of the operator space can be found. For illustration purposes, we
estimate the ground energy of the Bose-Hubbard Hamiltonian
\cite{childs_bosehubbard_2014} defined as:
\begin{equation}\label{eq:BHHamiltonianApp}
    \hat H = -J \sum_{0 \leq i < j \leq m} (\hat a_i^\dagger \hat a_j + \hat a_j^\dagger \hat a_i)
    - \frac{U}{2} \sum_{k = 1}^m \hat n_k(\hat n_k - 1),
\end{equation}
where $J > 0$ is the hopping parameter and $U\geq0$ is the many-body interaction
parameter. To characterize the Hamiltonian, we use the dimensionless parameter
$\Lambda = \frac{nU}{J}$. It was shown that in the non-interacting limit
(${\Lambda \rightarrow 0}$)---the so-called \emph{superfluid regime}---the
ground state of the Bose-Hubbard Hamiltonian can be produced by mean of a linear
optical network from a Fock basis state. Conversely, computational capabilities
beyond passive linear optics are required to prepare the ground state in
general, for instance using the nonlinear Kerr gates
\cite{yalouz_encoding_2021}. In the rest, we show that in spite of this, the
classical shadow protocol we introduce allows one to estimate the ground energy
efficiently in the interacting regime ($\Lambda > 0$). The following
\cref{lm:BHDecomposition}, that exhibits the decomposition of Bose-Hubbard
Hamiltonian along the isotypic subspaces of $\Ls(\Hc_m^n)$, will be useful.

\begin{lemma}\label{lm:BHDecomposition}
    Let $\hat H = \hat H_{\text{hop}} + \hat H_{\text{int}}$ be the Bose-Hubbard
    Hamiltonian as defined in \cref{eq:BHHamiltonianApp} where $\hat
    H_{\text{hop}}$ ($\hat H_{\text{int}}$) denotes the hopping (interacting)
    term. Then, it admits the following decomposition along the irreducible
    subspaces of $\Ls(\Hc_m^n)$:
\begin{equation}
    \begin{aligned}
        \Pi_{\lambda_0}(\hat H) & = -\frac{Un(n-1)}{m+1}\Id, \\
        \Pi_{\lambda_1}(\hat H) & = \hat H_{\text{hop}}, \\
        \Pi_{\lambda_2}(\hat H) & = \hat H_{\text{int}} - \Pi_{\lambda_0}(\hat H).
    \end{aligned}
\end{equation}
\end{lemma}

\begin{proof}
    First, by a counting argument we have 
    \begin{equation}
        \tr{\hat H} 
            = \tr{\hat H_{\text{int}}} 
            = -\frac{Un(n-1)}{m+1}\binom{n+m-1}{n},
    \end{equation}
    from what the expression of $\Pi_{\lambda_0}(\hat H)$ follows. Next, we
    decompose the interaction Hamiltonian $\hat H_{\text{int}}$. Using the
    canonical commutation relations $[\hat a_i, \hat a_j^\dagger] =
    \delta_{i,j}$, we find that     
    \begin{equation}
        \hat H_{\text{int}} 
            = -\frac{U}{2}\sum_{i=1}^m \hat n_i(\hat n_i - 1) 
            = -\frac{U}{2}\sum_{i = 1}^m \hat a_i^{\dagger 2}\hat a_i^2,
    \end{equation}
    from what follows that 
    \begin{equation}
        \Pi_{\lambda_0}(\hat H_{\text{int}}) = \Pi_{\lambda_0}(\hat H).
    \end{equation}

    To find the projection of $\hat H_{\text{int}}$ into $\Hc_{\lambda_1}$, we
    first exhibit a basis of $\Ls(\Hc_m^1) \simeq \Hc_{\lambda_0} \oplus
    \Hc_{\lambda_1}$. We use this basis to express $\Bc_{m, 1}$ and characterize
    the projection via \cref{pr:individualProj}. A convenient basis for the
    algebra of linear optical Hamiltonian consists of the following $m^2$
    operators \cite{parellada_lie_2026}:
    \begin{equation}
        \begin{aligned}
            \Xc = \aco\ \frac{1}{\sqrt{2}}\lpr \hat a_i^\dagger \hat a_j + \hat a_j^\dagger \hat a_i\ \rpr\acf_{1 \leq i<j\leq m},\ 
            \Yc = \aco\ \frac{\imath}{\sqrt{2}}\lpr \hat a_i^\dagger a_j - \hat a_j^\dagger \hat a_i\ \rpr\acf_{1 \leq i<j\leq m},\ 
            \Zc = \aco\ \hat n_i\ \acf_{i=1}^m.
        \end{aligned}
    \end{equation}
    Using the identity $\hat n_m = n- \sum_{i=1}^{m-1} \hat n_i$, i.e., elements
    of $\Zc$ are not linearly independent, we obtain $\dim \Hc_{\lambda_1} =
    m^2-1$ independent operators for a basis of $\Hc_{\lambda_1}$ by removing
    e.g., $\hat n_m$ from $\Zc$. That is to say, 
    \begin{equation}
        \Bc_{m, 1} = \Xc \cup  \Yc \cup \Zc - \{\hat n_m\}.
    \end{equation}

    The operators associated with nonzero coefficient $\expval*{\hat c_{\bm k,
    \bm \ell}}_{\hat H_{\text{int}}}$ in \cref{eq:dualCoefficient} are the
    diagonal operators, i.e. in $\Zc$. Observe that for $\hat n_i \in \Zc -
    \{\hat n_m\}$, $\expval*{\hat c_{\bm k, \bm \ell}}_{\hat H_{\text{int}}}$ is
    a constant that depends on $m$ and $n$ but independent of $i$, from what
    follows that $\Pi_{\lambda_0}(\hat H_{\text{int}}) +
    \Pi_{\lambda_1}(\hat H_{\text{int}}) \propto \Id$. This further implies that
    $\Pi_{\lambda_1}(\hat H_{\text{int}}) = 0$ as $\Hc_{\lambda_0}$ is spanned by the
    identity and the claim follows.
\end{proof}

\begin{figure}[!t]
    \centering
    \subfloat[\justifying Size $N$ of the classical shadow required to get an
    approximation of the ground energy to within $5\%$ for values of ${m = 2, 4,
    6, 8}$. For each plot, the line represents the average value of $N$ over 20
    runs. For clarity, only the subset of runs with size less than 600 steps are
    represented as points.\label{fig:BHSizeShadow}]{
        \includegraphics[width=\textwidth]{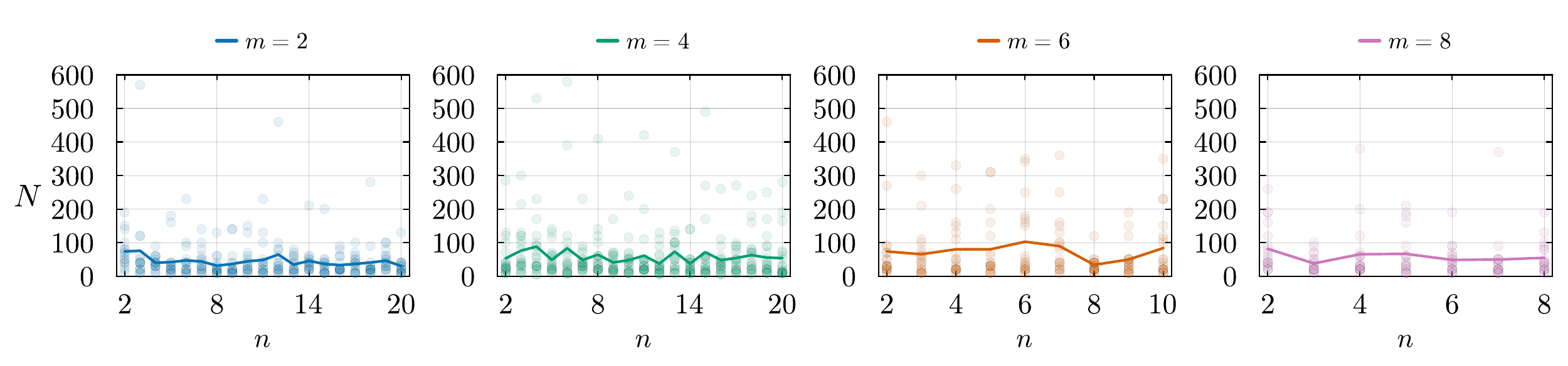}
    }
    \hfill
    \subfloat[\justifying Convergence rate of estimating the ground state
    energy for $\Lambda = 6$. \label{fig:BHConvergenceRate}]{
        \includegraphics[width=.8\textwidth]{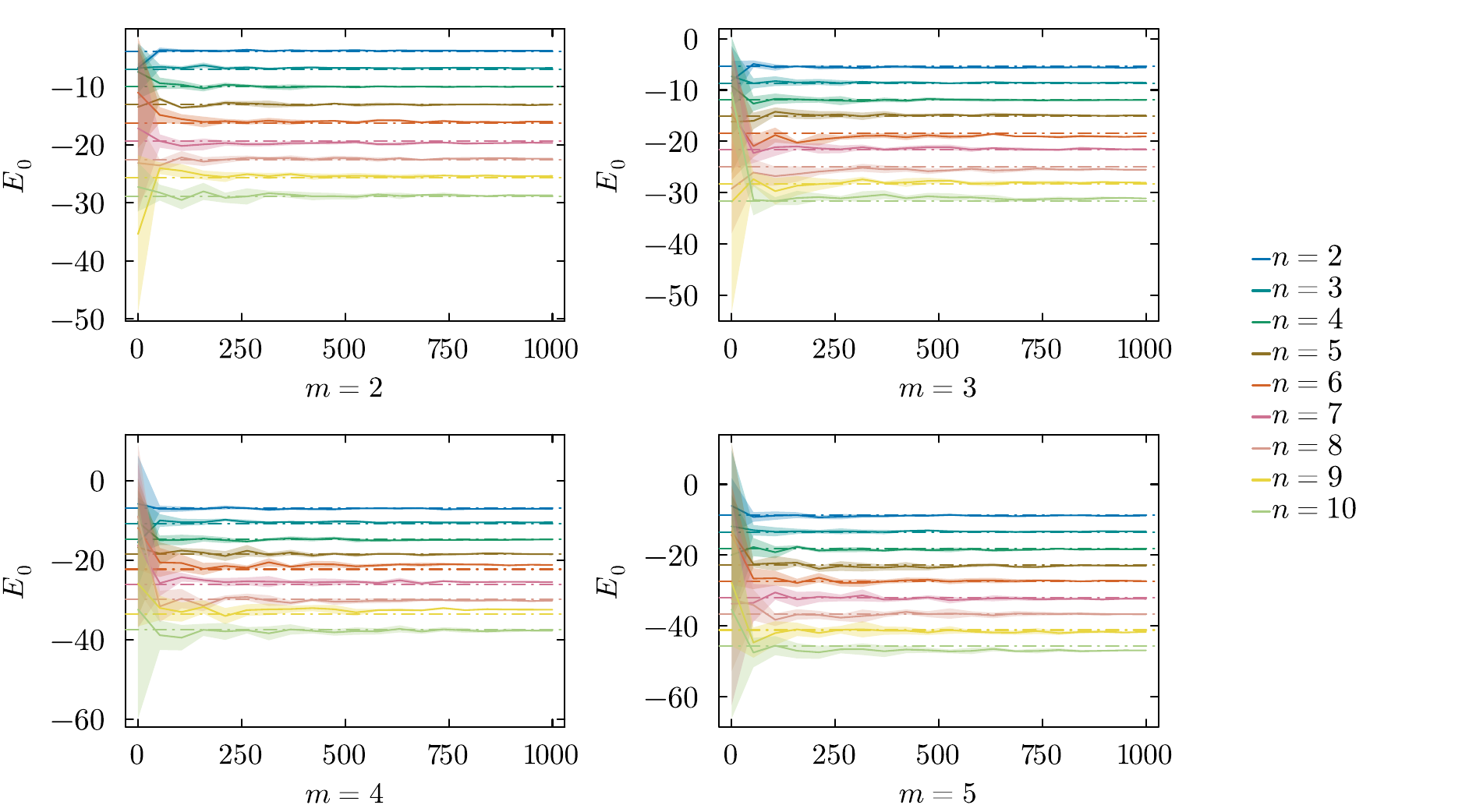}
    }

    \caption{\justifying \textbf{Resource estimation for estimating ground energy.}
    We fix $\Lambda = 6$, i.e., the ground state cannot be produced by linear optics acting on a Fock state alone. The numerical results suggest that the energy can be efficiently estimated with classical shadows.
    \cref{fig:BHSizeShadow}: Shadow size for estimating ground energy. Interestingly, as the degree of the observable is constant throughout, the size $N$ of the classical shadow required to estimate the ground energy within $5\%$ seems to be independent of the number of modes and photons.
    \cref{fig:BHConvergenceRate}: Convergence rate and distribution of the
    estimation. }
    \label{fig:prefectBHsim}
\end{figure}

\subsection{Numerical experiments}

For the numerical experiments, the input state is either obtained by exact
diagonalization of the Hamiltonian (perfect scenario) or produced from imperfect
single-photons (noisy scenario). In both setting the classical shadow protocol
is performed on the said input state and the collection $\{(U_i, \bm
s_i)\}_{i=1}^N$ is returned for different values of $N$. Then, the ground energy
estimate
\begin{equation}
    E_0 \approx \frac{1}{N}\sum_{i=1}^N \expval{\Ms^{(n)^{-1}}(\hat H)}_{\varphi_m^n(U_i^\dagger)\ket{\bm s_i}}
\end{equation}
is obtained via the method exposed in \cref{pr:degreeKObs} and the
decomposition shown in \cref{lm:BHDecomposition}. That is to say, the density
matrices $\hat \rho_{(U_i, \bm s_i)}$ are never computed. As expected, the main
bottleneck to performing numerical simulation is the data-collection phase,
which rapidly becomes intractable.

Above, we proved that states produced by computational capabilities beyond
passive linear optical transformations acting on Fock basis states could as well
be learnt. Such states could be produced by e.g., adaptive measurements or via
non-linear optical gates. Here, we aim to provide numerical evidences supporting
this claim. For the numerical experiments reported in
\cref{fig:prefectBHsim}, we fix $\Lambda = \frac{nU}{J} = 6$. That is to say, the ground
state cannot be produced by linear optics. Nonetheless, our results suggest that
meaningful information (i.e., estimate of the ground energy) can indeed be
recovered. To run the numerical simulations, we exactly diagonalise the
Hamiltonian and use the vector associated with the lowest energy as input
state-vector for the classical simulation. In a real-life scenario, the ground
state would be produced by a machine with capabilities beyond linear optics
\cite{yalouz_encoding_2021}.

\subsection{Experimental implementation}

We experimentally demonstrate the shadow tomography protocol on the task of
estimating the ground energy of the Bose-Hubbard Hamiltonian on a
twenty-four-mode, twelve-photon universal photonic processor. As our machine
implements arbitrary quadratic Hamiltonians, we implement the ground state of
the Bose-Hubbard Hamiltonian in the superfluid regime ($U=0$). We perform
experiments of various sizes, and give an overview of the architecture in
\cref{fig:BelenosForBH}. Considering the overhead in circuit size induced by
the resource-preparation (see below) and the pseudo-PNR steps (see
\cref{sec:ppnr}), an $m$-site, $n$-photon experiment effectively requires $M
= n(m + 1) - 1$ modes. The ground state can be prepared by a cascade of
beam-splitters acting on the Fock input state $\ket{n, 0^{\otimes m-1}}$ as
described in \cite{yalouz_encoding_2021}. The resource state $\ket{n}$ is
obtained by from $n$ single-photons via a QFT interferometer and post-selecting
on the first $n-1$ modes being occupied by the vacuum--giving a success
probability of $p_{\ket{n}} = \frac{n!}{n^n}$, which remains not too small for
the values of $n$ considered for these experiments. Because of the probabilistic
nature of the resource state preparation step, we sampled $10^7$ and $10^8$
shots for the $m=3$ and $4$-site experiments respectively, where a shot is an
event with as least one detector clicking.

\arxiv{}{
\begin{figure}[!ht]
    \includegraphics{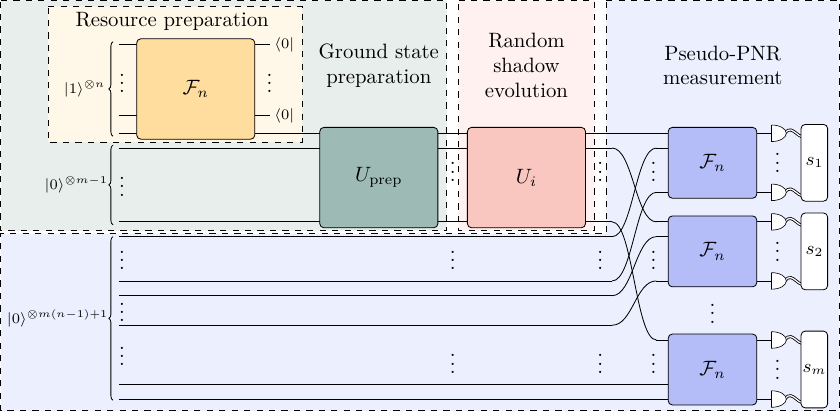}
    \caption{\justifying \textbf{Illustration of the circuit implemented on
    Belenos}. The resource preparation step consists in preparing the Fock state
    $\ket{n}$ from $n$ single-photons via a QFT interferometer and
    post-selecting on the first $n-1$ modes being occupied by the vacuum.
    $U_{\text{prep}}$ implements the ground state via a cascade of
    beam-splitters. The remainder of
    the circuit, which consists of a random evolution followed by pseudo-PNR
    measurement, mimics that the experimental runs on Ascella.}
    \label{fig:BelenosForBH}
\end{figure}}

\end{document}